\documentclass[letterpaper]{IEEEtran}
\usepackage{graphicx}
\usepackage{xcolor}

\usepackage{amsmath}
\usepackage{amssymb}
\usepackage{cite}
\usepackage[colorlinks=Flase]{hyperref}
\usepackage{bbm}
\usepackage{algorithm}
\usepackage{multirow} 

\usepackage{float}
\floatstyle{plaintop}
\restylefloat{table}
\usepackage{algpseudocode}
\usepackage{caption}
\usepackage{steinmetz}
\usepackage{subcaption}
\usepackage{amsthm}
\usepackage[english]{babel}
\newtheorem{theorem}{Theorem}
\newtheorem{proposition}{Proposition}
\newtheorem{lemma}{Lemma}
\makeatletter
\newenvironment{myproof}[1][\proofname]{\par
    \normalfont \topsep6\p@\@plus6\p@\relax
    \trivlist
    \item[\hskip\labelsep
    \itshape
    #1\@addpunct{.}]\ignorespaces
}{%
\endtrivlist\@endpefalse
}
\makeatother
\newcounter{relctr} 
\everydisplay\expandafter{\the\everydisplay\setcounter{relctr}{0}} 
\DeclareMathOperator*{\minimize}{minimize} 
\DeclareMathOperator*{\maximize}{maximize} 
\usepackage[figurename=Fig.]{caption}
\newcommand{\norm}[1]{\left\lVert#1\right\rVert}
\usepackage{xpatch}

\newtheoremstyle{remarkstyle}%
  {}
  {}
  {\itshape}
  {}
  {\itshape}
  {.}
  {.5em}
  {}

\theoremstyle{remarkstyle}
\newtheorem{remark}{Remark}

\newcommand\labelrel[2]{%
  \begingroup
    \refstepcounter{relctr}%
    \stackrel{\textnormal{(\alph{relctr})}}{\mathstrut{#1}}%
    \originallabel{#2}%
  \endgroup
}
\AtBeginDocument{\let\originallabel\label} 

\ifCLASSINFOpdf
\else
\fi

\begin{document}

\title{Beamforming and Waveform Optimization for RF Wireless Power Transfer with Beyond Diagonal Reconfigurable Intelligent Surfaces}

\author{Amirhossein~Azarbahram,~\IEEEmembership{Graduate~Student~Member,~IEEE,}
        Onel~L.~A.~López,~\IEEEmembership{Senior~Member,~IEEE,}
        Bruno~Clerckx,~\IEEEmembership{Fellow,~IEEE,}
       Marco~Di~Renzo,~\IEEEmembership{Fellow,~IEEE,}
        and~Matti~Latva-Aho,~\IEEEmembership{Fellow,~IEEE}
\thanks{This work is partially supported by the Research Council of Finland (Grants 348515, 362782, and 346208 (6G Flagship)); by the European Commission through the Horizon Europe/JU SNS project Hexa-X-II (Grant no. 101095759); by the Finnish-American Research \& Innovation Accelerator. The work of M. Di Renzo was supported in part by the Nokia Foundation, the French Institute of Finland, and the French Embassy in Finland under the France-Nokia Chair of Excellence in ICT. The work of M. Di Renzo was also supported in part by the European Commission through the Horizon Europe Project titled COVER under Grant 101086228, the Horizon Europe Project titled UNITE under Grant 101129618, the Horizon Europe Project titled
INSTINCT under Grant 101139161, and the Horizon Europe Project titled
TWIN6G under Grant 101182794; in part by the Agence Nationale de
la Recherche (ANR) through the France 2030 Project titled ANR-PEPR
Networks of the Future under Grant NF-PERSEUS 22-PEFT-004; in part by
the CHIST-ERA Project titled PASSIONATE under Grant CHIST-ERA-22-
WAI-04 and Grant ANR-23-CHR4-0003-01; and in part by the Engineering
and Physical Sciences Research Council (EPSRC), part of U.K. Research and Innovation, and the U.K. Department of Science, Innovation and Technology through the CHEDDAR Telecom Hub under Grant EP/X040518/1 and Grant EP/Y037421/1, and through the HASC Telecom Hub under Grant EP/X040569/1. The work of B.~Clerckx is partially supported by UKRI grant EP/Y004086/1, EP/X040569/1, EP/Y037197/1, EP/X04047X/1, and EP/Y037243/1.}
\thanks{A.~Azarbahram, O.~L\'opez, and M.~Latva-Aho are with Centre for Wireless Communications (CWC), University of Oulu, Finland, (e-mail: \{amirhossein.azarbahram, onel.alcarazlopez, matti.latva-aho\}@oulu.fi). B.~Clerckx is with the Department of Electrical and
Electronic Engineering, Imperial College London, SW7 2AZ London, U.K, (e-mail:b.clerckx@imperial.ac.uk). M. Di Renzo is with Universit\'e Paris-Saclay, CNRS, CentraleSup\'elec, Laboratoire des Signaux et Syst\`emes, 3 Rue Joliot-Curie, 91192 Gif-sur-Yvette, France. (marco.di-renzo@universite-paris-saclay.fr), and with King's College London, Centre for Telecommunications Research -- Department of Engineering, WC2R 2LS London, U.K (marco.di\_renzo@kcl.ac.uk).

}%
}

\maketitle

\begin{abstract}

Radio frequency (RF) wireless power transfer (WPT) is a promising technology to seamlessly charge low-power devices, but its low end-to-end power transfer efficiency remains a critical challenge. To address the latter, low-cost transmit/radiating architectures, e.g., based on reconfigurable intelligent surfaces (RISs), have shown great potential. Beyond diagonal (BD) RIS is a novel branch of RIS offering enhanced performance over traditional diagonal RIS (D-RIS) in wireless communications, but its potential gains in RF-WPT remain unexplored. Motivated by this, we analyze a BD-RIS-assisted single-antenna RF-WPT system to charge a single rectifier, and formulate a joint beamforming and multi-carrier waveform optimization problem aiming to maximize the harvested power. We propose two solutions relying on semi-definite programming for fully connected BD-RIS, a successive convex approximation (SCA)-based beamforming approach, and an efficient low-complexity iterative method relying on SCA. Numerical results show that the proposed algorithms converge and that adding transmit sub-carriers or RIS elements improves the harvesting performance. We show that the transmit power budget impacts the relative power allocation among different sub-carriers depending on the rectifier's operating regime, while BD-RIS shapes the cascade channel differently for frequency-selective and flat scenarios. Finally, we verify by simulation that BD-RIS and D-RIS achieve the same performance under pure far-field line-of-sight conditions (in the absence of mutual coupling). Meanwhile, BD-RIS outperforms D-RIS as the non-line-of-sight components of the channel become dominant.

\end{abstract}


\IEEEpeerreviewmaketitle


\section{Introduction}

\IEEEPARstart{F}{uture} wireless communication systems must ensure seamless green connectivity among numerous low-power devices. For this, it is essential to mitigate electronic waste resulting from battery replacements and reduce disruptions caused by battery depletion \cite{intro1, intro2, lópez2023highpower}. Energy harvesting (EH) technologies are fundamental enablers for this by providing wireless charging capability and promoting sustainable Internet of Things \cite{3gppamb, ZEDHEXA}.  EH devices may harvest energy from readily available or dedicated sources. However, the former, a.k.a., ambient EH, may not be possible in all environments and/or may call for large device form factors \cite{intro3}, while the latter is supported by wireless power transfer (WPT) technologies such as inductive coupling, laser power beaming, magnetic resonance coupling, and radio frequency (RF) radiation. Hereinafter, we focus on RF-WPT, which can provide multi-user wireless charging capability over large distances while using the wireless communications infrastructure.


\subsection{Preliminaries}

A WPT system comprises three key building blocks: i) energy transmitter (ET); ii) wireless channel; and iii) energy receiver (ER). The charging signal is generated and amplified using a direct current (DC) power source at the ET. Then, it is upconverted to RF and transmitted over the wireless channel. The ER collects the signal and converts it to DC for EH. Each block includes some power consumption and loss sources, which bring new challenges to the system design. In fact, a key challenge of WPT systems is their inherently low end-to-end power transfer efficiency, which is given by
\begin{equation}\label{eq:wptbase}
    e = \underbrace{\frac{P^{t}_{rf}}{P^{t}_{dc}}}_{e_1} \times \underbrace{\frac{P^{r}_{rf}}{P^{t}_{rf}}}_{e_2} \times \underbrace{\frac{P^{r}_{dc}}{P^{r}_{rf}}}_{e_3} = \frac{P^{r}_{dc}}{P^{t}_{dc}}.
\end{equation}
Herein, $e_1$ is impacted by the power consumption sources at the transmitter side, such as high-power amplifier, while the channel losses, i.e., shadowing and fading, can decrease $e_2$. Note that the RF-to-DC conversion inefficiency at the receiver side influences $e_3$. Thus, it is important to address these and provide novel solutions to maximize efficiency. 

Waveform optimization is required to enhance efficiency at the ER. In fact, multi-tone waveforms can leverage the ER's non-linearity and enhance the performance in terms of DC harvested power \cite{clerckx-nonlinearharvester, clerckx2018beneficial}. Also, beamforming techniques can compensate for wireless channel losses by pointing the transmit signal effectively toward the ERs' direction. The transmitter architecture determines the type of beamforming, which can be either active or passive. An active beamforming leverages active antenna elements connected to dedicated RF chains, while a passive beamforming utilizes low-cost nearly passive elements. For instance, a fully digital transmitter comprises active antenna elements, each with a dedicated RF chain, leading to high cost but offering the best beamforming gain. On the other hand, analog architectures relying, e.g., on phase shifters, can reduce the cost but also degrade the performance. There exist hybrid architectures, which combine the two mentioned architectures, leading to a reduced number of RF chains and a tradeoff between cost/complexity and performance \cite{hybridbeamsurvey}. 

Traditional analog and hybrid architectures still need to cope with the challenges caused by complex analog networks. There are novel transmit architectures that avoid this additional complexity and provide beamforming gains at a much lower cost. For instance, dynamic metasurface antennas (DMA) utilize a limited number of RF chains, each one feeding multiple antenna elements on a waveguide \cite{DMAWPT,MyEBDMA}. Moreover, reconfigurable intelligent surfaces (RIS) have become popular to enhance performance in wireless systems, especially in the presence of blockages/obstacles. This emerging technology may rely on planar surfaces comprising nearly passive scattering elements that introduce amplitude and phase changes to incident electromagnetic waves. Therefore, reflective-type RIS can smartly tune the reflected signal and point it toward the desired direction, providing significant beamforming gains \cite{RIS-basis}. In addition to reflective-type RIS, there exist some novel RIS-based architectures such as active RIS providing extra capabilities at the transmitter \cite{active_ris}, or reflective-refractive models \cite{reflect_refract_RIS}. 

A recent generalization of diagonal RIS (D-RIS) is given by beyond diagonal RIS (BD-RIS), which is characterized by scattering matrices not constrained to be diagonal, thus elements/ports interconnected via tunable impedances \cite{bdris_main_ref}. The optimality of BD-RIS in lossless designs was proved in \cite{BD-RIS_scattering_clerckx}. BD-RIS, through reconfigurable inter-connections, enables impinging waves to flow through the surfaces, hence a new degree of control and higher flexibility to manipulate waves by forming fully-, group-, tree-, and forest-connected structures \cite{BD-RIS_scattering_clerckx, BD-RIS_General}. The former provides the highest flexibility in shaping the reflected signal, though lower complexities based on tree- and forest-connected architectures are attractive to achieve the performance-complexity Pareto frontier \cite{BD-RIS_General, Pareto_Frontier_nerini}. Note that decreasing the number of connections between the scattering elements reduces both the performance gains and complexity, showcasing an interesting performance-complexity trade-off. BD-RIS is a general framework enabling multiple modes, such as reflective, hybrid, transmissive and reflective, and multi-sector, while encompassing other RIS architectures as special cases \cite{Multi_Sector_BDRIS, TR_REF_BDRIS}.

\subsection{Prior Works}

Waveform and beamforming optimization for WPT with non-linear EH has received significant research attention, although most works focus on traditional fully digital transmitters. For instance, the authors of \cite{BFRFDCsingletone} propose transmit and receive beamforming solutions for a fully digital multiple-input multiple-output (MIMO) system to increase the DC harvested power, demonstrating the superiority of RF over DC combining. The authors of \cite{clreckxWFdesign} investigate the design of WPT waveforms for different EH models and show that waveforms designed considering nonlinear EH yield substantial improvements in harvested DC power. Interestingly, in \cite{brunolowcomp}, the authors propose a low-complexity channel-adaptive waveform design for single-ER WPT systems, leading to near-optimal EH performance. In  \cite{azarbahram2024deepreinforcementlearningmultiuser}, a low-complexity near-optimal beamforming is proposed for maximizing the weighted sum harvested DC power in multi-ER scenarios. Moreover, the waveform and beamforming optimization problem for large-scale multi-antenna WPT systems has been investigated in \cite{largescale}, demonstrating the effectiveness of large-scale WPT deployments in improving the power transfer efficiency. The authors of \cite{jointWFandBFMIMOclreckx} leverage beamforming and optimized multi-sine waveforms in a MIMO WPT system and show that the joint optimization leads to significant gains compared to beamforming-only designs in terms of harvested power. All mentioned studies considered just the EH non-linearity, while \cite{SISOAllClreckx} considers both the power amplifier and rectifier's non-linearity. Therein, it is shown that the waveform design considering power amplifier non-linearity leads to significant gains compared to those focusing only on the rectifier.  Moreover, multiple works, e.g., \cite{OnelRadioStripes, onellowcomp, azarbahram2023radio}, focus on RF power maximization by considering linear EH for fully digital WPT.

Novel low-cost WPT transmit architectures for reducing the implementation cost/complexity is another topic gaining attention recently. The authors of \cite{azarbahram2024waveform} study the waveform and beamforming optimization for DMA-assisted WPT systems and show that DMA can be used to meet the DC requirement of the ERs with considerably fewer RF chains compared to a fully-digital counterpart. The authors of \cite{hybrid_SWIPT} propose a minimum-power beamforming design for simultaneous wireless information and power transfer (SWIPT) systems with a hybrid transmit architecture. The authors of \cite{ris_rf_power} derive the number of RIS elements in an RIS-aided WPT system required to outperform traditional WPT relying on active antenna elements. Interestingly, the harvested power maximization problem is addressed in \cite{clerckxRISWPT,ris_swipt_nonlin} for RIS-aided WPT and SWIPT systems with non-linear EH, respectively. Therein, the authors highlight the extra beamforming gains on the harvested DC power provided by RISs, although only focusing on conventional D-RIS. Note that BD-RIS architectures have been widely investigated for communication purposes \cite{BD-RIS_Graph, BD-RIS_General, BD-RIS_scattering_clerckx, bdris_main_ref, direnzoclerckxmutual}, while current research on RIS-aided WPT is still limited to conventional diagonal (local) RIS \cite{ris_swipt_nonlin, clerckxRISWPT}. Notably, another important yet unexplored topic is the use of BD-RIS for SWIPT systems with non-linear EH models. However, this work focuses on WPT alone, noting that maximizing the harvested DC power is first crucial as it directly impacts the achievable information transmission rate in SWIPT systems. 


\subsection{Contributions}

Efficient WPT system design and optimization requires accurate non-linear EH models \cite{clerckx2018beneficial, clreckxWFdesign}. However, most existing works employing non-linear EH primarily focus on traditional fully digital WPT designs. In contrast, this paper investigates BD-RIS-assisted WPT under non-linear EH models, a novel direction not yet explored.\footnote{This work is an extension of our conference paper \cite{myBDrisWCNC}, wherein we present a BD-RIS beamforming and waveform optimization algorithm. In the current work, we propose two other solutions for BD-RIS beamforming and an efficient iterative solution for waveform optimization. Moreover, we provide thorough algorithm performance comparisons and a comprehensive waveform behavior analysis.} Our main contributions are:
\par \textit{\textbf{First}}, we formulate a joint beamforming and waveform optimization problem for a single-input single-output (SISO) multi-carrier BD-RIS-aided WPT system to maximize the DC harvested power. For this, we assume there is no direct path between the ET and ER. Our problem requires novel solutions due to the following reasons: \\
i) The methods available in the literature for WPT with D-RIS \cite{clerckxRISWPT} do not apply to BD-RIS due to the additional mathematical complexity caused by the non-diagonal scattering matrix. Specifically, when optimizing BD-RIS, in contrast to D-RIS, all the elements of the scattering matrix can be non-zero, and a unitary constraint must be forced on the designed matrix, which significantly increases the mathematical complexity and calls for novel mathematical solutions.\\
ii) The methods available for optimizing BD-RIS, e.g., as in \cite{BD-RIS_Graph, BD-RIS_General}, which focus on RF power maximization, do not apply to the WPT case with non-linear EH. The reason is that the DC harvested power depends on the fourth moment of the signal in addition to the second moment, a.k.a., RF power. Specifically, RF power maximization only deals with the complexity caused by coupling between the waveform and beamforming. On the other hand, DC power maximization introduces inter-subcarrier dependencies in multi-carrier waveform and beamforming design, requiring novel methods to exploit EH non-linearity effectively \cite{clreckxWFdesign}.
\par \textit{\textbf{Second}}, we decouple the waveform optimization and beamforming problem via alternating optimization, and propose: \\
i) An efficient iterative method relying on successive convex approximation (SCA) and Karush–Kuhn–Tucker (KKT) conditions for the waveform optimization problem. Specifically, we prove that the objective function is convex w.r.t the waveform, hence we approximate it using its first-order Taylor approximation. The resulting problem is convex in the local neighborhood, and a closed-form solution can be obtained using KKT conditions and can then be iteratively updated.\\
ii) A beamforming optimization framework relying on SCA and semi-definite programming (SDP), which models the rank-1 constraint using the atomic and Frobenius norms \cite{shabir2024electromagneticallyconsistentoptimizationalgorithms}. Unlike semi-definite relaxation (SDR), this approach guarantees a rank-1 feasible solution without relaxation. \\
iii) SDR-based beamforming method, which relaxes the rank-1 constraint, but has lower time complexity compared to the one relying on SDP due to fewer approximations. The proposed SDP and SDR-based methods differ fundamentally from \cite{clerckxRISWPT}, as they address the added complexities introduced by the non-diagonal structure of the BD-RIS scattering matrix and the imposed unitary constraint. Notably, the SDP-based method further distinguishes itself by preserving the rank-1 constraint, in contrast to SDR-based approaches that relax this constraint.\\
iv) SCA-based beamforming, which leverages Taylor expansion to approximate the problem by dividing the unitary constraint into two separate regions. This algorithm has lower computational complexity than SDR, but yields a suboptimal solution due to the constraint relaxation.\\
v) An efficient iterative beamforming method by leveraging the Neumann series approximation, SCA, and KKT conditions. Specifically, we approximate the problem with local convex subproblems by leveraging the Neumann series. Then, we leverage the KKT conditions to solve each subproblem and use an iterative method using low-complexity closed-form expressions to update the optimization variables. The final solution is iteratively obtained using the SCA method. This method reduces the computational complexity from the fourth to the second degree of the number of RIS elements compared to the SDP and SDR-based approaches.
\par \textit{\textbf{Third}}, we verify by simulations that:\\
i) The proposed algorithms converge, with convergence time scaling with the number of sub-carriers and RIS elements. \\
ii) The SDR-based algorithm achieves the optimal performance for a single-carrier system, while it outperforms the proposed iterative method in terms of harvested DC power for multi-carrier systems.\\
iii) Increasing the number of sub-carriers or RIS elements improves the performance in terms of harvested DC power. \\
iv) The frequency selectivity of the channel impacts how BD-RIS shapes the cascade channel, leading to performance changes. Specifically, when the incident and reflective channels tend to be frequency-flat, the BD-RIS can shape the cascade channel more effectively since it affects all the sub-carriers similarly. On the other hand, in the frequency-selective case, the variation among sub-carriers in the cascade channel is greater because the BD-RIS configuration affects each sub-carrier differently. \\
v) The transmit power budget impacts the relative power allocation of the designed waveform depending on the input RF power to the rectifier. \\
vi) BD-RIS does not introduce any additional performance gains in terms of harvested DC power compared to D-RIS in far-field line-of-sight (LoS) channels (in the absence of mutual coupling) even with multi-carrier signals. On the other hand, in Rician channels with non-LoS (NLoS) components, BD-RIS outperforms D-RIS by effectively tuning the cascade channel, leading to increased harvested DC power. This extends the findings of \cite{BD-RIS_scattering_clerckx} to multi-carrier WPT.

The remainder of the paper is structured as follows. Section~\ref{sec:system} introduces the system model and problem formulation. The optimization framework for beamforming and waveform design is elaborated in Section~\ref{sec:OPTIMIZATION}. Section~\ref{sec:numerical} presents the numerical results, while Section~\ref{sec:conclude} concludes the paper.

\textbf{Notations:} Bold lower-case and upper-case letters represent vectors and matrices. The $\ell_2$-norm operator is denoted by $\norm{\cdot}$, while  $\norm{\cdot}_\star$ and $\norm{\cdot}_F$ are the atomic and Frobenius norm, respectively. $\Re\{\cdot\}$ and  $\Im\{\cdot\}$ denote the real and imaginary parts of a complex matrix or scalar, and $|\cdot|$ represents the absolute value operator. Moreover, $(\cdot)^T$, $(\cdot)^H$, and $(\cdot)^\star$ denote the transpose, transposed conjugate, and conjugate, respectively. Additionally, $[\cdot]_{i,l}$ denotes the element in the $l$th column and the $i$th row of a matrix. The vectorization operator is represented by $\mathrm{Vec}(\cdot)$, and its inverse is denoted by ${\mathrm{Vec}^{-1}_{D \times D}(\cdot)}$. Finally, $\mathbf{I}_D$ represents a $D \times D$ identity matrix and $\mathrm{diag}(\mathbf{a})$ is a diagonal matrix with its main diagonal being the elements of vector $\mathbf{a}$.

\section{System Model \& Problem Formulation}\label{sec:system}

We consider a SISO WPT system with $N$ sub-carriers aided by a BD-RIS consisting of $M$ elements. The adoption of multi-tone waveforms is motivated by their ability to exploit the rectifier's non-linearity to improve the performance \cite{clerckx2018beneficial}. For simplicity, we assume perfect channel state information at the ET as in \cite{clerckxRISWPT, bdris_main_ref, clreckxWFdesign}. Moreover, we assume there is no direct path between the ET and the ER, which is motivated by typical RIS applications to improve coverage in areas with blockages. By excluding the direct link, we isolate and highlight the contribution of the BD-RIS to the overall performance, allowing for a fair and focused comparison between different RIS architectures. The system model and the path of the incident and reflecting channels at the $n$th sub-carrier are illustrated in Fig.~\ref{fig:systemmodel}.a. This section describes the transmit and receive signals, BD-RIS, rectenna, and the problem formulation.

\begin{remark}
    The SISO setup is chosen to focus on the core contributions of this work, namely, the investigation of frequency diversity, BD-RIS beamforming gains, and nonlinear energy harvesting. This provides a clear and tractable framework for analyzing the behavior of BD-RIS systems and enables fair comparisons with conventional RIS designs. Importantly, the use of BD-RIS introduces significant mathematical complexity due to its general scattering structure, particularly when combined with nonlinear rectenna models and multi-carrier waveform optimization. Therefore, the considered SISO scenario serves as a necessary starting point for characterizing the fundamental gains of BD-RIS in WPT.
\end{remark}

\begin{figure}[t]
    \centering
    \includegraphics[width=0.6\columnwidth]{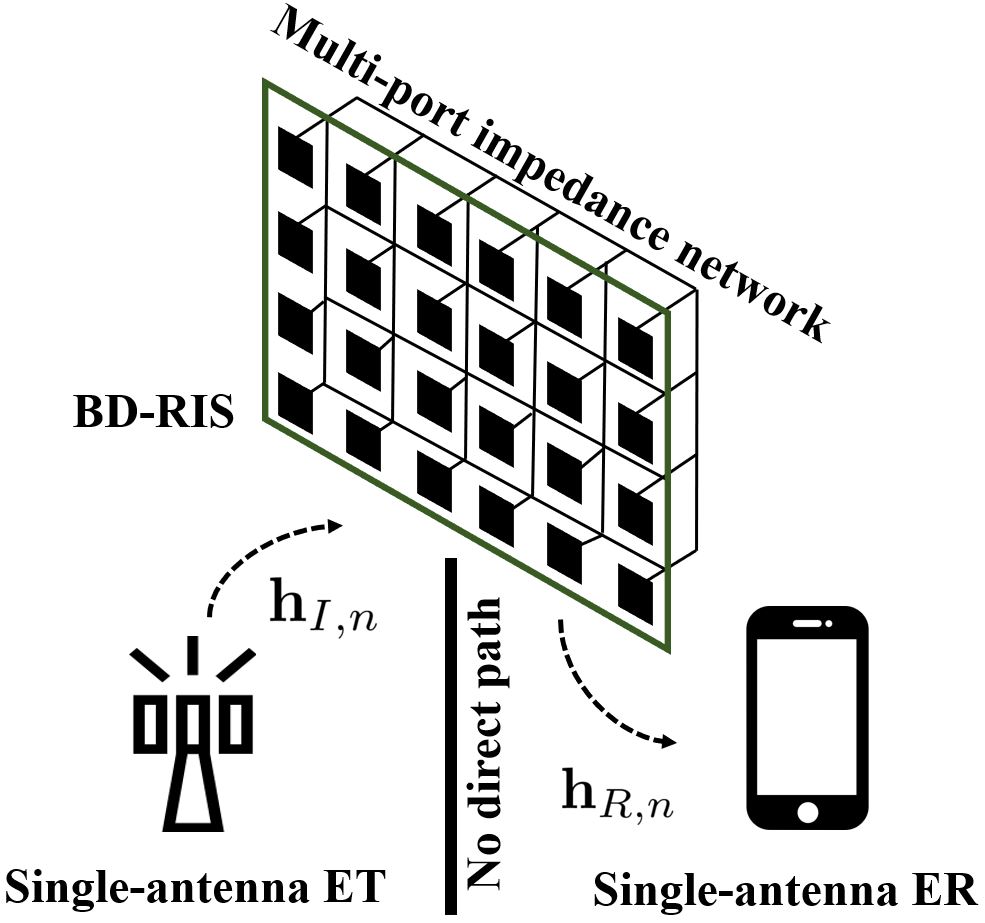} \\
    \vspace{2mm}
    \includegraphics[width=0.45\columnwidth]{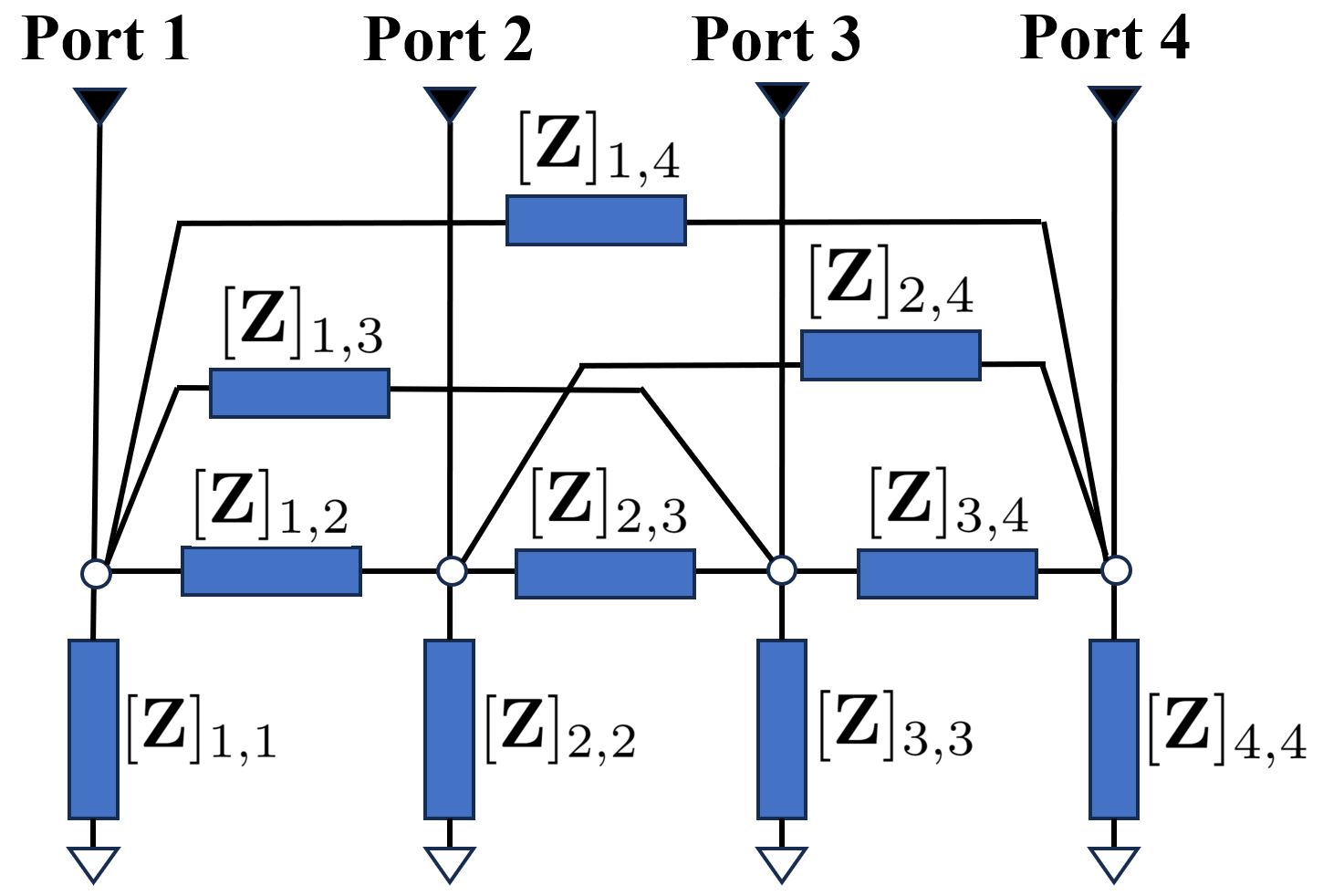}  
    \hspace{4mm}\includegraphics[width=0.45\columnwidth]{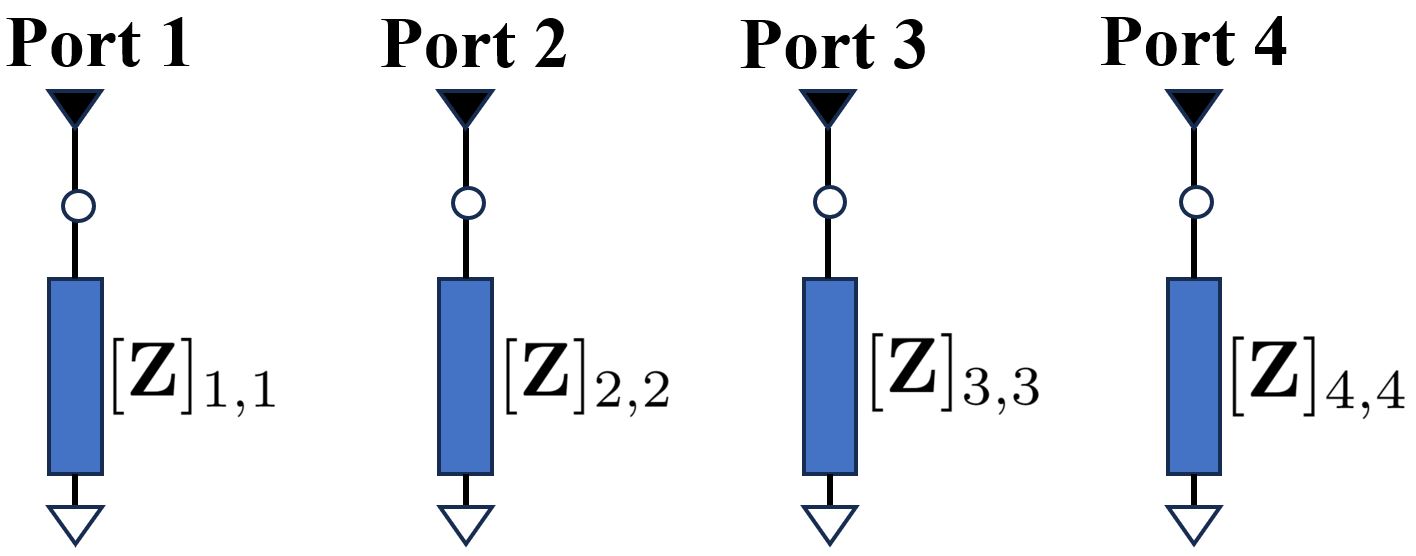}
    \caption{Illustration of (a) the system model comprising a single-antenna ET, a BD-RIS, and a single-antenna ER; and (b) examples of 4-port impedance networks for fully connected BD-RIS (bottom-left) and D-RIS (bottom-right).}
    \label{fig:systemmodel}
\end{figure}

    
\subsection{Transmit/Receive Signals and BD-RIS}

We consider $N$ sub-carriers with $f_n = f_c + (n - 1)\Delta f$ being the frequency of the $n$th sub-carrier, where $f_c$ is the lowest frequency and $\Delta f$ is the sub-carrier spacing. Thus, the transmit signal at time $t$ can be written as $\Re\bigl\{ \sum_{n = 1}^{N} s_n e^{j2\pi f_n t} \bigr\}$, where $s_n$ is the complex weight of the $n$th sub-carrier.\footnote{We assume a linear model for the power amplifier, being mainly interested in the impact of non-linear rectennas. Notably, this is a valid assumption for a proper transmit power such that the power amplifier does not enter the saturation region \cite{SISOAllClreckx}.} The transmit signal propagates through the wireless channel. We denote by $\mathbf{h}_{I,n} \in \mathbb{C}^{M\times 1}$ the channel between the transmitter and RIS and by $\mathbf{h}_{R,n} \in \mathbb{C}^{M\times 1}$ the channel between the RIS and receiver for the $n$th sub-carrier.

Assuming a multiport network model, each D-RIS element is modeled as a port connected to an independent tunable impedance \cite{BD-RIS_scattering_clerckx,univnerini,bdrissportdesign}. This leads to the so-called D-RIS model, which is mathematically characterized by diagonal phase-shift matrices. In BD-RIS, the matrix of tunable elements  $\mathbf{Z} \in \mathbb{C}^{M \times M}$ has non-zero off-diagonal entries \cite{BD-RIS_General}. Thus, the scattering matrix can be formulated as \cite{univnerini}
\begin{equation}\label{eq:scattering_matrix}
    \mathbf{\Theta} = (\mathbf{Z} + Z_0\mathbf{I}_M)^{-1}(\mathbf{Z} - Z_0\mathbf{I}_M),
\end{equation}
where $Z_0$ denotes the reference impedance. In the case of reciprocal and lossless BD-RIS,\footnote{For simplicity and since the goal is to investigate the gains of BD-RIS for WPT, we consider the case with no mutual coupling between different ports.} the scattering matrix is symmetric ($\mathbf{\Theta} = \mathbf{\Theta}^T)$ and unitary ($\mathbf{\Theta}^H\mathbf{\Theta} = \mathbf{I}_M$). Note that in the case of fully connected BD-RIS, all the elements in $\mathbf{Z}$ can be non-zero, while for other sub-architectures, some off-diagonal entries become zero, and for D-RIS, $\mathbf{Z}$ becomes diagonal. Fig.~\ref{fig:systemmodel}.b illustrates an example of a 4-port impedance network for fully connected BD-RIS and D-RIS.

The cascade channel at the $n$th sub-carrier is written as $h_n = \mathbf{h}_{R,n}^T\mathbf{\Theta}{\mathbf{h}_{I,n}}$,\footnote{The underlying general model and assumptions can be seen in \cite{bdrissportdesign}.} and the received signal at time $t$ is given by $y(t) = \sum_{n= 1}^{N} \Re\bigl\{ s_n h_n e^{j2\pi f_n t}\bigr\}$.



\subsection{Rectenna}

The ER model, i.e., rectenna, consists of the antenna equivalent circuit and a single-diode rectifier for transforming the RF received signal into DC \cite{azarbahram2024waveform, clerckx2018beneficial, clreckxWFdesign}. By leveraging this model, assuming perfect matching, and using Taylor's expansion, the output current of the rectenna is given by 
\begin{equation}\label{eq:dccurrent}
    i_{dc} =  \sum_{i\ even, i\geq2}^{\bar{n}} K_i \mathbb{E}\bigl\{y(t) ^i\bigr\},
\end{equation}
where $K_2$, $K_4$ are determined by the characteristics of the rectifier circuit. One can consider $\bar{n} = 4$ and model the main sources of non-linearity using the fourth order term, while the second order term represents the ideal linear rectenna model \cite{clreckxWFdesign, clerckxRISWPT}. Note that \eqref{eq:dccurrent} requires averaging over time samples, which makes the optimization extremely challenging due to the need to consider a high sampling rate for achieving sufficient accuracy. Fortunately, by expressing \eqref{eq:dccurrent} in the frequency domain and assuming $\bar{n} = 4$, we obtain \cite{clreckxWFdesign}
\begin{align}\label{eq:parseval}
    i_{dc} &= \frac{K_2}{2} \sum_{n}  |{s_n h_{n}}|^2 + 
    \frac{3K_4}{8} \sum_{\substack{n_0, n_1, n_2, n_3 \\ n_0 + n_1 = n_2 + n_3}} {({h}_{n_0} {s}_{n_0})}^\star \ldots\nonumber\\
    &\qquad\qquad\qquad\qquad ({h}_{n_1} {s}_{n_1})^\star{({h}_{n_2} {s}_{n_2})} {({h}_{n_3} {s}_{n_3})},
\end{align}
which is more tractable than the sampling-dependent model. Notice that $i_{dc}$ is a real-valued scalar obtained by averaging the second and fourth moments of the received real-valued signal.


\subsection{Problem Formulation}

We aim to maximize the harvested power given a transmit power budget, $P_T$, and assuming perfect channel state information (CSI). By leveraging the fact that the DC harvested power is an increasing function of the DC at the ER, the problem can be formulated as
\begin{subequations}\label{graph_problem}
\begin{align}
\label{graph_problem_a} \maximize_{s_n, \boldsymbol{\Theta}} \quad &  i_{dc} \\
\textrm{subject to} \label{graph_problem_c}  \quad & \boldsymbol{\Theta} = \boldsymbol{\Theta}^T, \\
\label{graph_problem_d}  \quad & \boldsymbol{\Theta}^H \boldsymbol{\Theta} = \mathbf{I}_M, \\
\label{graph_problem_e}  \quad & \frac{1}{2} \sum_{n = 1}^{N} |{s_n}|^2 \leq P_T.
\end{align}
\end{subequations}
Notably, \eqref{graph_problem} is a non-convex problem due to the unitary constraint, the $i_{dc}$ non-linearity, and the coupling between the variables. To cope with this, we decouple \eqref{graph_problem} into separate problems for optimizing waveform and beamforming and leverage alternating optimization to solve it. 

\begin{remark}
    Note that the power consumption of BD-RIS scales with the number of interconnections between its ports, which directly increases the circuit complexity and, consequently, the overall power consumption compared to traditional D-RIS architectures \cite{RIS_power_conusmption_TWC_2024}. However, this work focuses solely on the output power of the active transmitter, and thus, we assume that the BD-RIS has sufficient power available for its internal operation and impedance network tuning.
\end{remark} 
\begin{remark}
    For conventional diagonal RIS, two main CSI acquisition strategies exist \cite{CSI_survey_RIS}: i) channel estimation using partially activated RIS elements by RF chains, offering low overhead and compatibility with BD-RIS, but incurring additional hardware costs; ii) Cascaded channel estimation via RIS using pilot designs, which avoids RF chains but needs to be adapted for BD-RIS due to its tightly coupled channel structure. A BD-RIS-specific approach is explored in \cite{channel_estim_BDRIS}.
\end{remark} 

\section{Beamforming and Waveform Optimization}\label{sec:OPTIMIZATION}

Here, we provide optimization methods for solving the beamforming and waveform-related subproblems. First, we propose a low-complexity iterative method for waveform optimization problems relying on SCA and KKT conditions. Then, we leverage SDP, SDR, SCA, and KKT conditions to solve the beamforming subproblem using different approaches.


\subsection{Waveform Optimization with Fixed $h_n$}

First, we assume the cascade channel $h_n$ to be fixed and optimize the signal weights $s_n, \forall n$. We proceed by defining $s_n = \bar{s}_ne^{j\tilde{s}_n}$, where $\bar{s}_n$ and $\tilde{s}_n$ are the amplitude and the phase of $s_n$, respectively. Similarly, we can write $h_n = \bar{h}_ne^{j\tilde{h}_n}$. It is evident that the optimal $\tilde{s}_n$ must compensate for the phases in \eqref{eq:parseval}, leading to a real-valued $i_{dc}$ \cite{clerckx-nonlinearharvester}. Thus, it is sufficient to have $\tilde{s}_n^* = - \tilde{h}_n, \forall n$ leading to $\bar{s}_ne^{j\tilde{s}_n^*} h_n = \bar{s}_n\bar{h}_n, \forall n$. Now the only goal of the optimization is to find the optimal amplitudes for the signal weights at different sub-carriers. For this, the problem for a given $\boldsymbol{\Theta}$ can be reformulated  as
\begin{subequations}\label{fixedPS}
\begin{align}
\label{fixedPS_a} \maximize_{\bar{s}_n} \quad &  {i}_{dc} \\
\textrm{subject to} \label{fixedPS_b}  \quad & \frac{1}{2} \sum_{n = 1}^{N} \bar{s}_n^2 \leq P_T.
\end{align}
\end{subequations}
\begin{lemma}\label{theorem:1}
The DC current $i_{dc}$ in \eqref{eq:parseval} is convex w.r.t. $\bar{s}_n$.
\end{lemma}
\begin{proof}
    Note that ${y(t)}$ is linear w.r.t. $\bar{s}_n$. By leveraging the second-order convexity condition \cite{boyd2004convex} and the linearity of the mathematical average operator, it can be easily verified that $i_{dc}$ is convex w.r.t. $y(t)$. Then, $i_{dc}$ is convex w.r.t. $\bar{s}_n$ since the composition of an affine with a convex function is convex.
\end{proof}
According to Lemma~\ref{theorem:1}, problem \eqref{fixedPS} is not convex since it deals with the maximization of a convex function. However, the convexity of \eqref{fixedPS_a} leads to the fact that $\Tilde{i}_{dc}(\bar{s}_n , \bar{s}_n^{(l)}) \leq i_{dc}$, where $\Tilde{i}_{dc}(\bar{s}_n , \bar{s}_n^{(l)})$ is the first order Taylor's expansion of $i_{dc}$ at the local point $\bar{s}_n^{(l)}$, formulated as
\begin{align}
    \Tilde{i}_{dc}(\bar{s}_n , \bar{s}_n^{(l)}) &= i_{dc}\big|_{\bar{s}_n = \bar{s}_n^{(l)}} + \sum_{n = 1}^{N} g(\bar{s}_n^{(l)})(\bar{s}_n - \bar{s}_n^{(l)}),
\end{align}
where
\begin{multline}\label{eq:derivativetaylor}
   g(\bar{s}_n) = K_2 \bar{h}_{n}^2{\bar{s}}_n + \frac{3K_4}{2}\biggl[{\bar{h}_n}^4{\bar{s}_n}^3 + 
    2 \sum_{n_1 \neq n} {\bar{h}_n}^2{\bar{h}_{n_1}}^2{\bar{s}_{n_1}}^2\bar{s}_n+ \\
    \sum_{\substack{n_2, n_3 \\ n_2 + n_3 = 2n \\ n_2 \neq n_3}} \bar{h}_{n_2}\bar{h}_{n_3}\bar{h}_{n}^2\bar{s}_{n_2}\bar{s}_{n_3}\bar{s}_{n} + \\
    \sum_{\substack{n_1, n_2, n_3 \\ -n_1 + n_2 + n_3 = n \\ n \neq n_1 \neq n_2 \neq n_3}} \bar{h}_{n_1}\bar{h}_{n_2}\bar{h}_{n_3}\bar{s}_{n_1}\bar{s}_{n_2}\bar{s}_{n_3} \bar{h}_{n}
    \biggr].
\end{multline}
Hereby and by removing the constant terms (since they do not impact the optimization procedure), the problem can be transformed into a convex problem at the neighborhood of the initial point $\bar{s}_n^{(l)}$ as 
\begin{subequations}\label{fixedPS_approx}
\begin{align}
\label{fixedPS_approx_a} \minimize_{\bar{s}_n} \quad & \xi_1 = -\sum_{n = 1}^{N}  g(\bar{s}_n^{(l)})\bar{s}_n \\
\textrm{subject to} \label{fixedPS_approx_b}  \quad & \frac{1}{2} \sum_{n = 1}^{N} {\bar{s}_n}^2 \leq P_T.
\end{align}
\end{subequations}
Finally, the problem can be iteratively solved using standard convex optimization tools, e.g., CVX \cite{cvxref}.

However, utilizing standard solvers for the optimization problem in \eqref{fixedPS_approx} requires solving a quadratic program \cite{boyd2004convex} in each SCA iteration. Notably, the complexity of quadratic programs scales with a polynomial function of the problem size, while the degree of the polynomial mainly depends on the type of solver. Let us consider a simple solver based on the Newton method, which has $\mathcal{O}(n^3)$ complexity \cite{boyd2004convex}, where $n$ is the problem size, leading to $n$ scaling with the number of variables $N$. Therefore, solving problems such as \eqref{fixedPS_approx} using standard tools is not computationally efficient.

Let us proceed by writing the Lagrangian of \eqref{fixedPS_approx} as 
\begin{align}\label{eq:lagrangeweight}
    \mathcal{L}(\bar{s}_n) = - \sum_{n =1}^N g(\bar{s}_n^{(l)}) \bar{s}_n + {\lambda^{(l)}}\bigg(\frac{1}{2}\sum_{n = 1}^N {\bar{s}_n}^2 - P_T\bigg),
\end{align}
where $\lambda^{(l)}$ is the dual variable corresponding to \eqref{fixedPS_approx_b}. Then, we can write the derivative of \eqref{eq:lagrangeweight} w.r.t. $\bar{s}_n$ as
\begin{equation}\label{eq:derivsn}
    \frac{\partial\mathcal{L}(\bar{s}_n)}{\partial \bar{s}_n} = - g(\bar{s}_n^{(l)}) + \lambda^{(l)}\bar{s}_n, \quad \forall n.
\end{equation}
By setting \eqref{eq:derivsn} to zero, we attain 
\begin{equation}
    \label{eq:S_s_star}\bar{s}_n^* = \frac{g(\bar{s}_n^{(l)})}{\lambda^{(l)}}, \quad \forall n.
\end{equation}
Moreover, the optimal solution must utilize the whole transmit power budget, i.e., $\frac{1}{2} \sum_{n = 1}^{N} {\bar{s}_n}^2 = P_T$. Thus, by substituting \eqref{eq:S_s_star} into the latter equality, the dual variable $\lambda^{(l)}$ is given by 
\begin{equation}
    \label{eq:S_lamb_update}\lambda^{(l)} = \sqrt{\frac{1}{2P_T}\sum_{n = 1}^N {g(\bar{s}_n^{(l)})}^2}.
\end{equation}
Finally, we introduce an additional parameter $\rho_s \leq 1$ to control the convergence of the solution such that the sub-gradient update is given by 
\begin{equation}\label{eq:S_s_update}
    \bar{s}_n^{(l + 1)} = \bar{s}_n^{(l)} + \rho_s(\bar{s}_n^* - \bar{s}_n^{(l)}),
\end{equation}
which can be used to obtain the local solution of \eqref{fixedPS_approx}. Hereby, we can iteratively solve \eqref{fixedPS}, which leads to considerably lower computation complexity since closed-form expressions are utilized in each SCA iteration. Algorithm~\ref{alg:waveformsca} illustrates the proposed iterative method for obtaining $s_n, \forall n$. First, the scaled matched filter approach, proposed in \cite{brunolowcomp}, is used to initialize the signal weights such that
\begin{equation}\label{eq:SMF}
    s_{n} = e^{-j{{\tilde{h}}_{n}}} {{\bar{h}_{n}}^{\beta}}\sqrt{\frac{2P_T}{\sum_{n_0 = 1}^{N}{\bar{h}_{n_0}}^{2\beta}}}, \quad \forall n.
\end{equation}
Then, the optimization variables and the corresponding dual variable are updated iteratively until convergence.

\begin{theorem}\label{theo_conv_proof_WF}
    Algorithm~\ref{alg:waveformsca} converges to a stationary point of problem~\eqref{fixedPS}.
\end{theorem}
\begin{proof}
    The first-order Taylor approximation is a global under-estimator of a convex function, leading to $\Tilde{i}_{dc}(\bar{s}_n , \bar{s}_n^{(l)}) \leq i_{dc}(\bar{s}_n)$. By denoting $\bar{s}_n^*$ as the solution of problem \eqref{fixedPS_approx} at the lth iteration, we can write $\Tilde{i}_{dc}(\bar{s}_n^* , \bar{s}_n^{(l)}) \geq \Tilde{i}_{dc}(\bar{s}_n^{(l)}, \bar{s}_n^{(l)}) = i_{dc}(\bar{s}_n^{(l)})$. Combining this with the fact that $i_{dc}$ is convex and $\rho_s \leq 1$ in \eqref{eq:S_s_update}, we obtain that the sequence $\{i_{dc}(\bar{s}_n^{(l)})\}$ is monotonically increasing. Since $i_{dc}$ is continuous and bounded due to constraint \eqref{fixedPS_b}, $\lim_{l \rightarrow \infty} i_{dc}(\bar{s}_n^{(l)}) = i_{dc}^*$ exists. Therefore, the solution converges to a point that fulfills the KKT conditions of the original problem \eqref{fixedPS}, indicating that it is a stationary point of \eqref{fixedPS}.
\end{proof}

\begin{algorithm}[t]
	\caption{Iterative waveform optimization (IT-WF).} \label{alg:waveformsca}
	\begin{algorithmic}[1]
            \State \textbf{Input:} $\{h_n\}_{\forall n}$, $\rho_s$, $\upsilon$ \quad \textbf{Output:} $s_n^{(l)}$
             \State \textbf{Initialize:} Initialize $s_n^{(l)}, \forall n$ using \eqref{eq:SMF}, $\xi_1 =\infty$
            \Repeat
                \State \hspace{-2mm} $\xi_1^\star \leftarrow \xi_1$, compute $\lambda^{(l)}$ using \eqref{eq:S_lamb_update}
                \State \hspace{-2mm} Obtain $\bar{s}_n^{(l + 1)}, \forall n$ using \eqref{eq:S_s_star} and \eqref{eq:S_s_update}
                \State \hspace{-2mm} $s_n^{(l + 1)} \leftarrow \bar{s}_n^{(l + 1)}e^{-j\tilde{h}_n}, \forall n$ 
                \State \hspace{-2mm} Compute $\xi_1$ using \eqref{fixedPS_approx_a}, $s_n^{(l)} \leftarrow s_n^{(l + 1)}$, $l \leftarrow l + 1$
            \Until{$\norm{1 - {\xi_1^\star}/{\xi_1}}\leq \upsilon$}
                
\end{algorithmic} 
\end{algorithm}


\subsection{SDR-based Beamforming for Fully Connected BD-RIS}\label{sec:SDROPTIMIZATION}


Herein, we provide an optimization method for the beamforming problem given $s_n, \forall n$. Note that this approach only applies to fully connected BD-RIS since it directly deals with  $\mathbf{\Theta}$ as the optimization variable. Let us proceed by rewriting \eqref{graph_problem} for fixed $s_n$ as
\begin{align}
\label{scatter_probelm} \maximize_{\boldsymbol{\Theta}} \quad &  i_{dc} \\
\textrm{subject to} \quad & \eqref{graph_problem_c}, \eqref{graph_problem_d}, \nonumber
\end{align}
which is highly complicated and non-convex due to the unitary constraint and the $i_{dc}$ non-linearity.

\begin{proposition}\label{theorem:vecotorization2}
    The cascade channel, i.e., $h_n = \mathbf{h}_{R,n}^T\mathbf{\Theta}{\mathbf{h}_{I,n}}$, can be rewritten as 
    \begin{equation}\label{eq:hnreform2}
         h_n = \mathbf{a}_n^T \boldsymbol{\theta},
    \end{equation}
    where $\mathbf{a}_n =  \mathbf{P}^T\mathrm{Vec}(\mathbf{h}_{I,n}\mathbf{h}_{R,n}^T) \in \mathbb{C}^{M(M + 1)/2 \times 1}$. Moreover, $\boldsymbol{\theta} \in \mathbb{C}^{M(M + 1)/2 \times 1}$ is the vector containing the lower/upper-triangle elements in $\mathbf{\Theta}$ and $\mathbf{P} \in \{0, 1\}^{M^2\times M(M + 1)/2}$ is a permutation matrix such that
\begin{equation}\label{eq:permute_matrix} [\mathbf{P}]_{M(m - 1) + n, k} = 
    \begin{cases}
        1, & k = {m(m - 1)}/{2} + n,\  1 \leq n \leq m 
        \\
        1, & k = {n(n - 1)}/{2} + m,\  m < n \leq M \\
        0, & \text{otherwise}.
    \end{cases}
\end{equation} 
\end{proposition}
\begin{myproof}
    The proof is provided in Appendix~\ref{appen3}
\end{myproof}

Leveraging Proposition~\ref{theorem:vecotorization2} allows us to remove the constraint \eqref{graph_problem_c} by using $\boldsymbol{\theta}$ as the optimization variable. However, the complexity caused by the unitary constraint still remains. For this, we leverage the idea in \cite{clerckxRISWPT} and define $\mathbf{z}_n = s_n\mathbf{a}_n$, $\mathbf{D}_0 = \mathbf{z}_1\mathbf{z}_1^H + \ldots + \mathbf{z}_N\mathbf{z}_N^H$, $\mathbf{D}_1 = \mathbf{z}_1\mathbf{z}_2^H + \ldots + \mathbf{z}_{N-1}\mathbf{z}_N^H$, and $\mathbf{D}_{N - 1} = \mathbf{z}_1\mathbf{z}_N^H$. Hereby, $i_{dc}$ can be reformulated as 
\begin{multline}\label{eq:idc_reform1}
    i_{dc} = \frac{1}{2} K_2 \boldsymbol{\theta}^H \mathbf{D}_0 \boldsymbol{\theta} + \frac{3}{8} K_4 \boldsymbol{\theta}^H \mathbf{D}_0 \boldsymbol{\theta} (\boldsymbol{\theta}^H \mathbf{D}_0 \boldsymbol{\theta})^H  \\ + \frac{3}{4} K_4 \sum_{n = 1}^{N - 1} \boldsymbol{\theta}^H \mathbf{D}_n \boldsymbol{\theta} (\boldsymbol{\theta}^H \mathbf{D}_n \boldsymbol{\theta})^H.
\end{multline}
Next, we need to formulate the constraint \eqref{graph_problem_d} as a function of the new optimization variable $\boldsymbol{\theta}$. For this, let us proceed by defining a permutation matrix $\mathbf{P}_i$, 
which extracts the $i$th row of $\boldsymbol{\Theta}$ from $\boldsymbol{\theta}$. Hereby, \eqref{graph_problem_d} can be rewritten as 
\begin{align}\label{eq:permutetheta}
    (\mathbf{P}_i \boldsymbol{\theta})^H(\mathbf{P}_j \boldsymbol{\theta}) = \mathrm{Tr}(\boldsymbol{\theta}\boldsymbol{\theta}^H \bar{\mathbf{P}}_{i, j}) =  \begin{cases}
        1, & i = j, 
        \\
        0, & i \neq j,
    \end{cases}
\end{align}

\noindent where $\bar{\mathbf{P}}_{i, j} = \mathbf{P}_i^T \mathbf{P}_j$ and $\mathbf{P}_i$ is a permutation matrix containing the $(iM - M + 1)$th to the $iM$th row of $\mathbf{P}$.

Now, we define $d_n = \boldsymbol{\theta}^H \mathbf{D}_n \boldsymbol{\theta}$, $\mathbf{d} = [ d_1, \ldots, d_N]^T$, and positive semidefinite matrices $\mathbf{K}_0 = \mathrm{diag}\{\frac{3}{8} K_4, \frac{3}{4} K_4, \ldots, \frac{3}{4} K_4\} \succeq 0$ and $\mathbf{X} = \boldsymbol{\theta} \boldsymbol{\theta}^H$. Hereby, the problem can be reformulated as 
\begin{subequations}\label{SDP_problem_gen}
\begin{align}
\label{SDP_problem_gen_a} \maximize_{\mathbf{d}, \mathbf{X} \succeq 0} \quad &  \frac{1}{2} K_2 d_0 +  \mathbf{d}^H \mathbf{K}_0 \mathbf{d} \\
\textrm{subject to} \label{SDP_problem_gen_c}  \quad & \mathrm{Tr}(\mathbf{X} \bar{\mathbf{P}}_{i, j}) = 1,\ \forall i = j, \\
\label{SDP_problem_gen_f}  \quad & \mathrm{Tr}(\mathbf{X} \bar{\mathbf{P}}_{i, j}) = 0,\ \forall i \neq j, \\
\label{SDP_problem_gen_d}  \quad & d_n = \mathrm{Tr}(\mathbf{X} \mathbf{D}_n),\ \forall n, \\
\label{SDP_problem_gen_e}  \quad & \mathrm{rank}(\mathbf{X}) = 1.
\end{align}
\end{subequations}

Note that problem \eqref{SDP_problem_gen} is still non-convex and challenging to solve since it deals with the maximization of a convex objective function and includes a rank-1 constraint. For \eqref{SDP_problem_gen_a}, SCA can be used to iteratively update the objective function by approximating it using its first-order Taylor expansion. Specifically, the quadratic term in \eqref{SDP_problem_gen_a} can be approximated in the neighborhood of $\mathbf{d}^{(l)}$ by \cite{clerckxRISWPT}
\begin{equation}\label{eq:taylor_sdp}
    f(\mathbf{d}, \mathbf{d}^{(l)}) = 2 \Re \bigl\{ {\mathbf{d}^{(l)}}^H \mathbf{K}_0 \mathbf{d} \bigr\} - {\mathbf{d}^{(l)}}^H \mathbf{K}_0 {\mathbf{d}^{(l)}}. 
\end{equation}

Since $f(\mathbf{d}, \mathbf{d}^{(l)}) \leq \mathbf{d}^H \mathbf{K}_0 \mathbf{d}$ always holds, it can be used as a lower bound, and maximizing Taylor's approximation iteratively leads to maximizing the original quadratic term. Note that \eqref{SDP_problem_gen_e} can be equivalently written as 
\begin{equation}\label{eq:rank1atomic}
    \norm{\mathbf{X}}_\star - \norm{\mathbf{X}}_F \leq 0,
\end{equation}
which is non-convex since it deals with the difference of two convex functions. However, the right-hand-side term can be approximated using its first-order Taylor expansion at a local point $\bar{\mathbf{X}}$. Specifically, $\norm{\mathbf{X}}_F$ can be replaced by a convex lower-bound given by \cite{shabir2024electromagneticallyconsistentoptimizationalgorithms}
\begin{equation}\label{eq:apporx_X}
    f(\mathbf{X}, \mathbf{\bar{X}}) = \norm{\bar{\mathbf{X}}}_F + \Re\biggl\{\sum_{i,l} \frac{[\mathbf{\bar{X}}]_{i,l}}{\norm{\bar{\mathbf{X}}}_F}\bigl([\mathbf{X}]_{i,l} - [\mathbf{\bar{X}}]_{i,l}\bigr)^\star \biggr\}.
\end{equation}
Hereby, \eqref{SDP_problem_gen} can be reformulated as 
\begin{subequations}\label{SDP_problem_rank}
\begin{align}
\label{SDP_problem_rank_a} \minimize_{\mathbf{X}} \quad & \Omega = \mathrm{Tr}(\mathbf{K}_1\mathbf{X}) \\
\textrm{subject to} \label{SDP_problem_rank_c}  \quad  &\norm{\mathbf{X}}_\star -  f(\mathbf{X}, \mathbf{\bar{X}}) \leq 0, \\
\label{SDP_problem_rank_e}  \quad & \mathbf{X} \succeq 0, \\
\quad & \eqref{SDP_problem_gen_c}, \eqref{SDP_problem_gen_f}, \nonumber
\end{align}
\end{subequations}
where $\mathbf{K}_1 = \mathbf{J} + \mathbf{J}^H$ and
\begin{equation}
    \mathbf{J} = -\frac{K_2}{4} \mathbf{D}_0 - \frac{3K_4}{8} d_0^{(l)} \mathbf{D}_0 - \frac{3 K_4}{4} \sum_{n = 1}^{N - 1} {d_n^{(l)}} \mathbf{D}_n.
\end{equation}
This problem is a standard SDP that can be solved using convex optimization tools such as CVX \cite{cvxref}. However, this approach requires many iterations for convergence since $\mathbf{X}$ is a large matrix and the approximation is done in both the objective function and a constraint. To cope with this complexity, we use SDR to reformulate the problem in the neighborhood of the initial point $\mathbf{d}^{(l)}$ as
\begin{subequations}\label{SDP_problem}
\begin{align}
\label{SDP_problem_a} \minimize_{\mathbf{X}} \quad & \Omega = \mathrm{Tr}(\mathbf{K}_1\mathbf{X}) \\
\textrm{subject to} \quad &  \eqref{SDP_problem_gen_c}, \eqref{SDP_problem_gen_f}, \eqref{SDP_problem_rank_e}, \nonumber
\end{align}
\end{subequations}
which is also a standard SDP similar to \eqref{SDP_problem_rank}. Moreover, if the obtained $\mathbf{X}^*$ is a rank-1 matrix, the SDR is tight and $\mathbf{X}^*$ is a stationary point of the original convex subproblem including rank-1 constraint, leading to the solution extracted by $\mathbf{X}^* = \boldsymbol{\theta}^* {\boldsymbol{\theta}^*}^H$. Then, $\boldsymbol{\theta}^*$ can be used to obtain $\mathbf{\Theta}^* = \mathrm{Vec}^{-1}(\mathbf{P}\boldsymbol{\theta}^*)$. However, it might happen that $\mathrm{rank}(\mathbf{X}^*) > 1$, which is the case that $\mathbf{X}^*$ is a stationary point of problem \eqref{SDP_problem} (see \cite{clerckxRISWPT} for the proof). For this, we obtain an approximate $\boldsymbol{\theta}^*$ using the Gaussian randomization method in \cite{guss_randomization}. Note that in this case, using $\boldsymbol{\theta}^*$ to construct $\mathbf{\Theta}$ leads to a symmetric $\mathbf{\Theta}$, but not necessarily a unitary matrix. Thus, it is essential to map the final solution into the feasible space of the problem. Let us proceed by writing $\mathbf{\Theta}' = \mathrm{Vec}^{-1}(\mathbf{P}\boldsymbol{\theta}^*)$. Then, by leveraging the fact that $\mathbf{\Theta}'$ is symmetric, we can write the singular value decomposition (SVD) as $\mathbf{\Theta}' = \mathbf{Q}\boldsymbol{\Sigma}\mathbf{Q}^T$, where $\mathbf{Q}$ is a unitary matrix and $\boldsymbol{\Sigma}$ is a diagonal matrix containing the singular values of $\mathbf{\Theta}'$. It is evident that if the diagonal elements of $\boldsymbol{\Sigma}$ are unit modulus,  $\mathbf{\Theta}'$ is unitary. However, this is only the case when $\mathrm{rank}(\mathbf{X}^*) = 1$ and SDR is tight, while for higher-rank cases, we propose a randomization-based method to obtain a feasible solution.

\begin{algorithm}[t]
	\caption{SDR-based beamforming and waveform optimization for fully connected BD-RIS (SDR-BDRIS).} \label{alg:SDPRIS}
	\begin{algorithmic}[1]
            \State \textbf{Input:} $\upsilon$, $\beta$, $\mathbf{h}_{R,n}, \mathbf{h}_{I,n}, \forall n$ \quad \textbf{Output:} $s_n^{(l)}, \forall n$, $\mathbf{\Theta}^*$
            \State \textbf{Initialize:} 
            \State $f^\star =0$, $i_{dc} = \infty$ \label{alg1:line:init_start}, $[\boldsymbol{Z}_{i,m}] = jZ_0, \forall i,m$, 
            \State $\mathbf{\Theta}^{(l)} = (\mathbf{Z} + Z_0\mathbf{I}_M)^{-1}(\mathbf{Z} - Z_0\mathbf{I}_M)$ 
            \State Compute $\mathbf{P}$, $h_n$, and $s_n^{(l)}$ using \eqref{eq:permute_matrix}, \eqref{eq:hnreform2}, and \eqref{eq:SMF}
            \State Compute $d_n^{(l)} = \boldsymbol{\theta}^{(l)}\mathbf{D}_n{\boldsymbol{\theta}^{(l)}}^H$, where $\boldsymbol{\theta}^{(l)} = \mathbf{P}^{-1}\mathrm{Vec}(\mathbf{\Theta}^{(l)})$
            \Repeat\label{algBDRIS:line:alter_start}
                \State \hspace{-2mm} $\Omega  =\infty$, $i_{dc}^\star \leftarrow i_{dc}$
                \Repeat\label{algBDRIS:passstart}
                    \State \hspace{-2mm} $\Omega ^\star \leftarrow \Omega $, solve \eqref{SDP_problem} to obtain $\mathbf{X}$
                    \State \hspace{-2mm} Compute $\boldsymbol{\theta}^{(l)}$ using Gaussian randomization
                    \State \hspace{-2mm} $d_n^{(l)} = \boldsymbol{\theta}^{(l)}\mathbf{D}_n{\boldsymbol{\theta}^{(l)}}^H$, $\mathbf{\Theta}^{(l)} = \mathrm{Vec}^{-1}(\mathbf{P}\boldsymbol{\theta}^{(l)})$
                    \State \hspace{-2mm} $l \leftarrow l + 1$
                \Until{$\norm{1 - {\Omega ^\star}/{\Omega }}\leq \upsilon$}\label{algBDRIS:passend}
                \State \hspace{-2mm} $h_n = \mathbf{h}_{R,n}^T\mathbf{\Theta}^{(l)}{\mathbf{h}_{I,n}}$ and run Algorithm~\ref{alg:waveformsca} to update $s_n^{(l)}$
                \State \hspace{-2mm} Compute $i_{dc}$ using \eqref{eq:parseval}
            \Until{$\norm{1 - i_{dc}^\star/ i_{dc}}\leq \upsilon$}\label{algBDRIS:line:alter_end}
            \State Run Algorithm~\ref{alg:PS_obtain} to obtain $\mathbf{\Theta}^*$
\end{algorithmic} 
\end{algorithm}

\begin{algorithm}[t]
	\caption{Randomization-based method for obtaining a feasible $\mathbf{\Theta}$.} \label{alg:PS_obtain}
	\begin{algorithmic}[1]
            \State \textbf{Input:} $\mathbf{\Theta}$, $K$, $\mathbf{h}_{R,n}, \mathbf{h}_{I,n}, s_n, \forall n$ \quad \textbf{Output:} $\mathbf{\Theta}^*$
             \State \textbf{Initialize:} Compute SVD of $\mathbf{\Theta}$ as $\mathbf{\Theta} = \mathbf{Q}\boldsymbol{\Sigma}\mathbf{Q}^T$, $i^\star_{dc} = 0$
            \For{$k = 1, \ldots, K$}
                \State \hspace{-2mm} Generate random $\phi_i \in [0, 2\pi], \forall i$ 
                \State \hspace{-2mm} Set $\boldsymbol{\phi} = [e^{j\phi_1}, \ldots, e^{j\phi_M}]^T$ and $\boldsymbol{\Sigma}' = \mathrm{diag}(\boldsymbol{\phi})$
                \State \hspace{-2mm} Compute $\mathbf{\Theta} = \mathbf{Q}\boldsymbol{\Sigma}'\mathbf{Q}^T$ and $i_{dc}$ using \eqref{eq:parseval}
                \State \textbf{if } $i_{dc} > i^\star_{dc}$, \textbf{ then } $i^\star_{dc} \leftarrow i_{dc}$, $\mathbf{\Theta}^* \leftarrow \mathbf{\Theta}$
            \EndFor
\end{algorithmic} 
\end{algorithm}

Algorithm~\ref{alg:SDPRIS} describes the proposed SDR-based method for beamforming and waveform optimization for fully connected BD-RIS. First, the optimization variables are initialized. Then, the waveform and scattering matrix are optimized in an alternative fashion through lines \ref{algBDRIS:line:alter_start}-\ref{algBDRIS:line:alter_end}. Specifically, beamforming is done by solving \eqref{SDP_problem} iteratively in lines \ref{algBDRIS:passstart}-\ref{algBDRIS:passend}, followed by iterative waveform optimization using Algorithm~\ref{alg:waveformsca}. Finally, Algorithm~\ref{alg:PS_obtain} is utilized to construct a feasible solution $\mathbf{\Theta}^*$ based on the characteristics of the obtained solution $\mathbf{\Theta}$. First, the SVD of $\mathbf{\Theta}$ is computed to obtain a unitary matrix $\mathbf{Q}$. Then, random phase shifts are generated for $K$ iterations to construct new $\boldsymbol{\Sigma}'$ matrices and their corresponding feasible solution $\mathbf{\Theta}$. Finally, the constructed $\mathbf{\Theta}$ with the best $i_{dc}$ is selected. Note that the SDP-BDRIS algorithm follows the same procedure as Algorithm~\ref{alg:SDPRIS}, but solves \eqref{SDP_problem_rank} instead of \eqref{SDP_problem}.

{The problems in \eqref{SDP_problem_rank} and \eqref{SDP_problem} are bounded due to the constraints in \eqref{SDP_problem_gen_c}, \eqref{SDP_problem_gen_f}, and \eqref{SDP_problem_rank_e}, along with the rank surrogate constraint \eqref{SDP_problem_rank_c} in the SDP-based formulation. The proposed SCA-based beamforming algorithms, based on SDP and SDR, iteratively solve convex approximations of the original non-convex problem \eqref{SDP_problem_gen}. At each iteration, both the objective and the non-convex surrogate rank constraint (in the SDP-based case) are approximated using first-order Taylor expansions that satisfy the standard SCA assumptions \cite{scutari2016paralleldistributedmethodsnonconvex, DC_constraint_SCA}. As in Theorem~\ref{theo_conv_proof_WF}, the SDP-based beamforming converges to a stationary point of the original beamforming problem \eqref{SDP_problem_gen}. For the SDR-based case, convergence to a stationary point of \eqref{SDP_problem_gen} holds only when the obtained solution of the relaxed problem is rank-1. Therefore, since Algorithm~\ref{alg:waveformsca} converges to a stationary point of the waveform subproblem \eqref{fixedPS}, the full SDP-BDRIS algorithm yields a stationary point of the joint problem \eqref{graph_problem}. On the other hand, SDR-BDRIS converges to a stationary point of \eqref{graph_problem} only when the solution to \eqref{SDP_problem} is rank-1. Otherwise, if the relaxed solution is of a higher rank, the rank-1 solution obtained via Algorithm~\ref{alg:PS_obtain} is a feasible approximate solution, and the resulting point is not guaranteed to satisfy the stationary conditions of the original problem \eqref{graph_problem}.}

\subsection{SCA-based Beamforming for Fully Connected BD-RIS}\label{sec:SCAOPTIMIZATION}

Here, we provide an optimization framework relying only on SCA by directly approximating \eqref{scatter_probelm}. Specifically, we directly deal with the matrix $\boldsymbol{\Theta}$ without lifting the variable size. Let us proceed by writing the first-order Taylor approximation of $i_{dc}$ w.r.t. $\boldsymbol{\Theta}$ as 
\begin{multline}
    \tilde{i}_{dc}(\boldsymbol{\Theta}, \boldsymbol{\Theta}^{(l)}) =  i_{dc}(\boldsymbol{\Theta}^{(l)}) \\ + 2\,\Re\!\left\{\operatorname{Tr}\!\left(\big(\nabla_{\boldsymbol{\Theta}} i_{dc}(\boldsymbol{\Theta}^{(l)})\big)^{\!H}\,(\boldsymbol{\Theta}-\boldsymbol{\Theta}^{(l)})\right)\right\}.
\end{multline}
Moreover, the unitary constraint can be divided into two separate constraints formulated as $\boldsymbol{\Theta}^H \boldsymbol{\Theta} \preceq  \mathbf{I}_M$ and $\boldsymbol{\Theta}^H \boldsymbol{\Theta} \succeq  \mathbf{I}_M$. By applying the Schur complement \cite{boyd2004convex}, the former convex constraint can be equivalently expressed as 
\begin{equation}\label{eq:schur}
    \begin{bmatrix}
        \mathbf{I}_M & \boldsymbol{\Theta}^H \\
        \boldsymbol{\Theta} & \mathbf{I}_M
    \end{bmatrix} \succeq 0,
\end{equation}
while the latter constraint is nonconvex. To enable tractable optimization, we approximate it via the
first-order Taylor expansion around the local point $\boldsymbol{\Theta}^{(l)}$, which leads to the
affine surrogate written as \cite{MM_survey}
\begin{equation}
    g({\boldsymbol{\Theta}^{(l)}}, {\boldsymbol{\Theta}}) = {\boldsymbol{\Theta}^{(l)}}^H\boldsymbol{\Theta}
    + \boldsymbol{\Theta}^H \boldsymbol{\Theta}^{(l)} - {\boldsymbol{\Theta}^{(l)}}^H\boldsymbol{\Theta}^{(l)}.
\end{equation}
Since $(\boldsymbol{\Theta}-\boldsymbol{\Theta}^{(l)})^H(\boldsymbol{\Theta}-\boldsymbol{\Theta}^{(l)}) 
\succeq 0$, this affine constraint provides a global inner approximation of 
the original constraint. Hence, the two constraints jointly enforce an approximate unitary constraint 
within the SCA framework. Note that enforcing these two constraints on the problem may lead to an empty 
feasible space, thus solver failure. To address this, we introduce 
a positive semidefinite auxiliary variable $\mathbf{S}$ in the affine surrogate constraint, 
while penalizing its trace in the objective function. This avoids solver failures by relaxing the unitary constraint to some extent, while the penalty pushes the auxiliary term towards zero, thereby encouraging meeting the constraint throughout the SCA iterations. Moreover, to avoid divergence when the local point 
$\boldsymbol{\Theta}^{(l)}$ is far from the unitary set, we introduce a trust region that restricts the update $\boldsymbol{\Theta}-\boldsymbol{\Theta}^{(l)}$ within a bounded neighborhood. 
These two strategies are standard in SCA to guarantee convergence \cite{scutari2016paralleldistributedmethodsnonconvex,boyd2004convex}. Hereby, the local convex problem in the neighborhood of $\boldsymbol{\Theta}^{(l)}$ can be written as 
\begin{subequations}\label{scatter_probelm_SCA}
\begin{align}
\label{scatter_probelm_SCA_a} \maximize_{\boldsymbol{\Theta}, \mathbf{S}\succeq 0} \quad & \Omega_{\textrm{sca}} = \tilde{i}_{dc}(\boldsymbol{\Theta}, \boldsymbol{\Theta}^{(l)}) - \sigma\mathrm{Tr}(\mathbf{S}) \\
\textrm{subject to} \quad & g({\boldsymbol{\Theta}^{(l)}}, {\boldsymbol{\Theta}}) + \mathbf{S} \succeq \mathbf{I}_M,\\
\quad & \norm{\boldsymbol{\Theta} - \boldsymbol{\Theta}^{(l)}}_F \leq \iota, \\
\quad & \eqref{graph_problem_c}, \eqref{eq:schur}.\nonumber
\end{align}
\end{subequations}

Algorithm~\ref{alg:scaSDPRIS} describes the proposed SCA-based method for beamforming and waveform optimization for fully connected BD-RIS. The overall procedure is similar to Algorithm~\ref{alg:SDPRIS}, while beamforming is optimized by iteratively solving \eqref{scatter_probelm_SCA}. Moreover, we use an adaptive penalty coefficient~$\sigma$, which is initialized with a small value to preserve feasibility in the early iterations and then gradually increased to enforce the approximate unitary constraint more strictly as the algorithm converges. Furthermore, to avoid infeasible solutions caused by relaxations, we map the final solution using Algorithm~\ref{alg:PS_obtain}. The proposed SCA-BDRIS algorithm includes a
non-decreasing sequence of the objective values since, in each iteration, the convex surrogate is solved optimally and is tight at the local point. Moreover, every obtained local solution satisfies the KKT conditions of the surrogate problem under standard assumptions \cite{scutari2016paralleldistributedmethodsnonconvex}. In case the obtained solution before mapping is unitary, the final solution is a stationary point of the original problem; otherwise, it is an approximate feasible solution after mapping.

\begin{algorithm}[t]
	\caption{SCA-based beamforming and waveform optimization for fully connected BD-RIS (SCA-BDRIS).} \label{alg:scaSDPRIS}
	\begin{algorithmic}[1]
            \State \textbf{Input:} $\upsilon$, $\sigma$, $\iota$, $\mathbf{h}_{R,n}, \mathbf{h}_{I,n}, \forall n$ \quad \textbf{Output:} $s_n^{(l)}, \forall n$, $\mathbf{\Theta}^*$
            \State \textbf{Initialize:} as in Algorithm~\ref{alg:SDPRIS}
            \Repeat\label{algscaBDRIS:line:alter_start}
                \State \hspace{-2mm} $\Omega_{\textrm{sca}}  =\infty$, $i_{dc}^\star \leftarrow i_{dc}$
                \Repeat\label{algscaBDRIS:passstart}
                    \State \hspace{-2mm} $\Omega_{\textrm{sca}} ^\star \leftarrow \Omega_{\textrm{sca}} $, solve \eqref{scatter_probelm_SCA} to obtain $\boldsymbol{\Theta}$
                    \State \hspace{-2mm} $l \leftarrow l + 1$, $\mathbf{\Theta}^{(l)} \leftarrow \boldsymbol{\Theta}$, $\sigma \leftarrow \min\{1, 1.5 \sigma \}$
                \Until{$\norm{1 - {\Omega_{\textrm{sca}} ^\star}/{\Omega_{\textrm{sca}} }}\leq \upsilon$}\label{algscaBDRIS:passend}
                \State \hspace{-2mm} $h_n = \mathbf{h}_{R,n}^T\mathbf{\Theta}^{(l)}{\mathbf{h}_{I,n}}$ and run Algorithm~\ref{alg:waveformsca} to update $s_n^{(l)}$
                \State \hspace{-2mm} Compute $i_{dc}$ using \eqref{eq:parseval}
            \Until{$\norm{1 - i_{dc}^\star/ i_{dc}}\leq \upsilon$}\label{algscaBDRIS:line:alter_end}
            \State Run Algorithm~\ref{alg:PS_obtain} to obtain $\mathbf{\Theta}^*$
\end{algorithmic} 
\end{algorithm}

\subsection{Iterative Beamforming Optimization}\label{sec:WFOPTIMIZATION}

Herein, we provide an alternative low-complexity optimization method for the beamforming problem, which applies to any BD-RIS structure since it deals with the $\mathbf{Z}$ matrix as the optimization variable. Let us proceed by defining $\mathcal{G}_\mathcal{Z}$ as the set of possible impedance matrices, which relies on the BD-RIS structure and inter-connections. Then, we rewrite \eqref{scatter_probelm} as
\begin{subequations}\label{fixedW}
\begin{align}
\label{fixedW_a} \maximize_{\mathbf{Z}} \quad &  i_{dc} \\
\textrm{subject to} \label{fixedWb} \quad &  \mathbf{Z} = \mathbf{Z}^T, \\
\label{fixedWc} \quad &  \mathbf{Z} \in \mathcal{G}_\mathcal{Z},
\end{align}
\end{subequations}
which is a highly non-linear non-convex problem due to the matrix inverse term in \eqref{eq:scattering_matrix}. Thus, the first challenge is to deal with this term, which appears in $h_n$. 

\begin{proposition}\label{theorem:2}
By leveraging the Neumann approximation and defining $\mathbf{\Omega}$ as a small increment to the impedance matrix $\mathbf{Z}$, such that $\norm{[\mathbf{\Omega}]_{i,m}} \leq \delta, \forall i,m$, we can write 
\begin{multline}\label{eq:neuman_impedance}
    (\mathbf{Z} + Z_0\mathbf{I}_M + \mathbf{\Omega})^{-1} \approx (\mathbf{Z} + Z_0\mathbf{I}_M)^{-1} - \\ (\mathbf{Z} + Z_0\mathbf{I}_M)^{-1}\mathbf{\Omega}(\mathbf{Z} + Z_0\mathbf{I}_M)^{-1},
\end{multline}
where $\delta$ must satisfy
\begin{equation}\label{eq:deltacondition}
    \delta \ll \frac{1}{\norm{(\mathbf{Z}^{(l)} + Z_0\mathbf{I}_M)^{-1}}_\infty}.
\end{equation}
\end{proposition}
\begin{myproof}
    The proof is provided in Appendix~\ref{appen2}.
\end{myproof}
By leveraging Proposition~\ref{theorem:2}, one can linearize $h_n$ using the Neumann approximation and iteratively improve the approximation. Thus, by defining an auxiliary variable $\mathbf{\Omega}$ as the increment to the impedance matrix $\mathbf{Z}^{(l)}$ at the $l$th iteration of the Neumann approximation, we can write
\begin{multline}\label{eq:hnapprox}
    h_n^{(l)} \approx \mathbf{h}_{R,n}^T(\mathbf{Z}^{(l)} + Z_0\mathbf{I}_M + \boldsymbol{\Omega})^{-1}(\mathbf{Z}^{(l)} - Z_0\mathbf{I}_M){\mathbf{h}_{I,n}} \\ 
    \approx \mathbf{h}_{R,n}^T\biggl[(\mathbf{Z} + Z_0\mathbf{I}_M)^{-1} -  \\(\mathbf{Z} + Z_0\mathbf{I}_M)^{-1}\mathbf{\Omega}(\mathbf{Z} + Z_0\mathbf{I}_M)^{-1}\biggr](\mathbf{Z}^{(l)} - Z_0\mathbf{I}_M){\mathbf{h}_{I,n}} \\
    = \mathbf{h}_{R,n}^T\mathbf{A}^{(l)}( \mathbf{I}_M-  
    \mathbf{\Omega}\mathbf{A}^{(l)})\mathbf{B}^{(l)}{\mathbf{h}_{I,n}} \\
    =  {\mathbf{a}_n^{(l)}}^T(\mathbf{I}_M - \mathbf{\Omega}\mathbf{A}^{(l)})\mathbf{b}_n^{(l)}, \quad \forall n,
\end{multline}
where $\mathbf{A}^{(l)} = (\mathbf{Z}^{(l)} + Z_0\mathbf{I}_M)^{-1}$, $\mathbf{a}_n^{(l)} = ({\mathbf{h}_n^R}^T\mathbf{A}^{(l)})^T$, $\mathbf{b}_n^{(l)} = \mathbf{B}^{(l)}{\mathbf{h}_{I, n}}$, and $\mathbf{B}^{(l)} =  \mathbf{Z}^{(l)} - Z_0\mathbf{I}_M$. 

By leveraging the linearization in \eqref{eq:hnapprox}, it can be easily verified that the objective in \eqref{fixedW} can be approximated using Taylor's expansion. This allows for the problem to be solved using convex optimization tools, similar to the previous optimization case. As mentioned earlier, one needs to solve a quadratic program at each iteration for that purpose. This leads to the problem size scaling with the symmetric matrix subspace with size $M^2$, making it computationally inefficient for large problems. To cope with this, we proceed by proposing an efficient iterative method to solve \eqref{fixedW}.

\begin{proposition}\label{theorem:vecotorization}
    Expression \eqref{eq:hnapprox} can be rewritten as 
    \begin{equation}\label{eq:hnreform}
         h_n = {\mathbf{a}_n^{(l)}}^T\mathbf{b}_n^{(l)} - {\mathbf{f}_n^{(l)}}^T\boldsymbol{\omega},
    \end{equation}
    where $\mathbf{f}_n^{(l)} =  \mathbf{P}^T \mathrm{Vec}(\mathbf{A}^{(l)}\mathbf{b}_n^{(l)}{\mathbf{a}_n^{(l)}}^T) \in \mathbb{C}^{M(M + 1)/2 \times 1}$. Moreover, $\boldsymbol{\omega} \in \mathbb{C}^{M(M + 1)/2 \times 1}$ is the vector containing the lower/upper-triangle elements in $\mathbf{\Omega}$ and $\mathbf{P} \in \{0, 1\}^{M^2\times M(M + 1)/2}$ is the permutation matrix defined in \eqref{eq:permute_matrix}.
\end{proposition}
\begin{proof}
By leveraging the same matrix properties provided in Appendix~\ref{appen3}, we proceed by rewriting 
\begin{align}
    h_n &\approx {\mathbf{a}_n^{(l)}}^T(\mathbf{I}_M\ - \mathbf{\Omega}\mathbf{A}^{(l)})\mathbf{b}_n^{(l)} \nonumber \\
    &= {\mathbf{a}_n^{(l)}}^T\mathbf{b}_n^{(l)} - \mathrm{Tr}(\mathbf{\Omega}\mathbf{A}^{(l)}\mathbf{b}_n^{(l)}{\mathbf{a}_n^{(l)}}^T) \nonumber \\ &=  {\mathbf{a}_n^{(l)}}^T\mathbf{b}_n^{(l)} - \mathrm{Vec}(\mathbf{A}^{(l)}\mathbf{b}_n^{(l)}{\mathbf{a}_n^{(l)}}^T)^T\mathrm{Vec}(\mathbf{\Omega}) \nonumber \\
    &=   {\mathbf{a}_n^{(l)}}^T\mathbf{b}_n^{(l)} - \mathrm{Vec}(\mathbf{A}^{(l)}\mathbf{b}_n^{(l)}{\mathbf{a}_n^{(l)}}^T)^T\mathbf{P}\boldsymbol{\omega},
\end{align}
matching \eqref{eq:hnreform}.
\end{proof}
\begin{lemma}\label{theorem:omegaconvex}
The $i_{dc}$ in \eqref{eq:parseval} is convex w.r.t. $\boldsymbol{\omega}$.
\end{lemma}
\begin{proof}
    It is evident from equation \eqref{eq:hnreform} that $h_n$ is affine w.r.t. $\boldsymbol{\omega}$. Furthermore, $i_{dc}$ is convex w.r.t. $h_n$, which can be easily proven using Lemma~\ref{theorem:1}. Thus, $i_{dc}$ w.r.t. $\boldsymbol{\omega}$ is a composition of an affine function and a convex function, which is convex.
\end{proof}

Proposition~\ref{theorem:vecotorization} and Lemma~\ref{theorem:omegaconvex} ensure that substituting \eqref{eq:hnreform} into \eqref{fixedW_a} leads to maximizing a convex objective function w.r.t. $\boldsymbol{\omega}$, which can be approximated similar to the waveform optimization procedure for fixed $h_n$. Specifically, the convex objective can be approximated using its first-order Taylor expansion at each local point and iteratively optimized. Thus, we proceed by writing the first-order Taylor's coefficient as
\begin{multline}\label{eq:taylor2}
     \mathbf{u}({\boldsymbol{\omega}}) = -K_2 \sum_n h_n\norm{s_n}^2 {\mathbf{f}_n^{(l)}}+ \\ -\frac{3K_4}{8} \sum_{\substack{n_0, n_1, n_2, n_3 \\ n_0 + n_1 = n_2 + n_3}}s^\star_{n_0}s^\star_{n_1}s_{n_2}s_{n_3}\biggl[h_{n_1}^\star h_{n_2}h_{n_3}{\mathbf{f}_{n_0}^{(l)}}^* + \\
    h_{n_0}^\star h_{n_2}h_{n_3}{\mathbf{f}_{n_1}^{(l)}}^* + 
    h_{n_0}^\star h_{n_1}^\star h_{n_3}{\mathbf{f}_{n_2}^{(l)}} + 
    h_{n_0}^\star h_{n_1}^\star h_{n_2}{\mathbf{f}_{n_3}^{(l)}}\biggr].
\end{multline}
For the sake of simplicity, let us denote $\mathcal{D}_\mathcal{G}$ as the set of indices in vector $\boldsymbol{\omega}$ that are zero due to the lack of connections between specific ports in the impedance network. By defining $\gamma \ll 1$ and setting 
\begin{equation}\label{eq:tau_update}
    \tau^{(l)} = {\gamma}/{\norm{(\mathbf{Z}^{(l)} + Z_0\mathbf{I}_N)^{-1}}_\infty}
\end{equation}
such that \eqref{eq:deltacondition} holds, and leveraging Proposition~\ref{theorem:vecotorization} to remove the constraint in \eqref{fixedWb}, \eqref{fixedW} can be transformed into a convex problem in the neighborhood of $\boldsymbol{\omega}^{(l)}$ as 
\begin{subequations}\label{fixedWreform}
\begin{align}
\label{fixedWreforma} \minimize_{\boldsymbol{\omega}} \quad & \xi_2 =-  {\mathbf{u}({\boldsymbol{\omega}^{(l)}})}^T \boldsymbol{\omega} \\
\textrm{subject to} \label{fixedWreformc}  \quad & \norm{[\boldsymbol{\omega}]_{r}} \leq \tau^{(l)}, \forall r, \\
  \label{fixedWreformd} \quad & [\boldsymbol{\omega}]_{r} = 0, r \in \mathcal{D}_\mathcal{G},
\end{align}
\end{subequations}
where \eqref{fixedWreformd} ensures that $\mathbf{\Omega}$ has the same zero entries as $\mathbf{Z}$ depending on the structure of BD-RIS. Note that for a fully connected BD-RIS, constraint \eqref{fixedWreformd} can be removed.

{Let us proceed by defining $\bar{M} = M(M + 1)/2$ and writing the Lagrangian of \eqref{fixedWreform} as
\begin{align}
    \mathcal{L}(\boldsymbol{\omega}) &= -\sum_{r = 1}^{\bar{M}} [\mathbf{u}^{(l)}]_r[\boldsymbol{\omega}]_r \nonumber  \\ &+ \sum_{r = 1}^{\bar{M}} \nu_r^{(l)} \biggl({\norm{[\boldsymbol{\omega}]_{r}}}^2 - {\tau^{(l)}}^2 \biggr) + \sum_{r \in \mathcal{D}_\mathcal{G}} \bar{\lambda}_r^{(l)}[\boldsymbol{\omega}]_{r}.
\end{align}
It is obvious that by forcing $[\boldsymbol{\omega}]_{r} = 0, \forall r \in \mathcal{D}_\mathcal{G}$, we have $\bar{\lambda}_r^{(l)}[\boldsymbol{\omega}]_{r} = 0, \forall r \in \mathcal{D}_\mathcal{G}$. Meanwhile, for $r\notin \mathcal{D}_\mathcal{G}$, the derivative of the Lagrangian function is given by
\begin{align}
    \frac{\partial\mathcal{L}(\boldsymbol{\omega})}{\partial[\boldsymbol{\omega})]_r} &=  -[\mathbf{u}^{(l)}]_r + 2\nu_r^{(l)} [\boldsymbol{\omega}]_r, \quad r\notin \mathcal{D}_\mathcal{G}.
\end{align}
We know that the solution must satisfy KKT conditions \cite{boyd2004convex}. Hence, by setting the derivative to zero, we can write the closed-form expression of the solution at the $l$th iteration as 
\begin{align}
    [\boldsymbol{\omega}^*]_r &= \begin{cases}
        0, & r \in \mathcal{D}_\mathcal{G}, \\
        {[\mathbf{u}^{(l)}]_r}/{ 2\nu_r^{(l)}}, & \text{otherwise},  
     \end{cases} \label{eq:Z_omega_star}.
\end{align}
Then, using the complementary slackness for the dual variable $\nu_r^{(l)}$, and leveraging \eqref{eq:Z_omega_star} and the fact that $\nu_r^{(l)} \geq 0, \forall r$, we have 
\begin{align}
    \nu_r^{(l)} &= \begin{cases}
        0, & r \in \mathcal{D}_\mathcal{G}, \\
        \norm{[\mathbf{u}^{(l)}]_r}/2{\tau^{(l)}}, & \text{otherwise}.  \label{eq:Z_nu_update}
        \end{cases}
\end{align}
Moreover, we introduce $\rho_\omega \leq 1$ to control the convergence of the solution, leading to writing the solution at the end of the $l$th iteration as
\begin{equation}\label{eq:Z_omega_update}
    \boldsymbol{\omega}^{(l + 1)} = \boldsymbol{\omega}^{(l)} + \rho_\omega(\boldsymbol{\omega}^* - \boldsymbol{\omega}^{(l)}),
\end{equation}
which can be leveraged to solve \eqref{fixedWreform} iteratively.} Moreover, it was verified that $\mathrm{Vec}(\mathbf{\Omega}^{(l + 1)}) = \mathbf{P}\boldsymbol{\omega}^{(l + 1)}$; thus, $\mathbf{\Omega}^{(l + 1)} = \mathrm{Vec}^{-1}_{M, M}(\mathrm{Vec}(\mathbf{\Omega}^{(l + 1)}))$ can be used for updating $\mathbf{Z}$ as
\begin{equation}\label{eq:Z_Z_update}
    \mathbf{Z}^{(l + 1)} = \mathbf{Z}^{(l)} + j\Im\{\mathbf{\Omega}^{(l + 1)}\},
\end{equation}
where only the imaginary part of $\mathbf{\Omega}^{(l + 1)}$ is updated to ensure that the impedance matrix, i.e., $\mathbf{Z}$, remains imaginary.

Algorithm~\ref{alg:BDRIS} illustrates the proposed iterative optimization method relying on alternating optimization, SCA, KKT conditions, and the Neumann approximation for BD-RIS-aided WPT (IT-BDRIS) given $\mathcal{D}_\mathcal{G}$. First, we initialize the variables such that the impedance matrix is set to the reference impedance value, while \eqref{eq:SMF} is used to initialize $s_n$. Then, the alternating optimization is done through lines~\ref{alg1:line:alter_start}-\ref{alg1:line:alter_end}. Specifically, each alternating optimization iteration consists of iteratively updating $\mathbf{Z}$ given fixed $s_n, \forall n$ in lines~\ref{alg1:line:digitSCA_start}-\ref{alg1:line:digitSCA_end}, followed by obtaining $s_n, \forall n$ using Algorithm~\ref{alg:waveformsca}. This procedure is repeated iteratively until convergence.

The solution of the iterative beamforming design obtained in Algorithm~\ref{alg:BDRIS} lines~\ref{alg1:line:digitSCA_start}-\ref{alg1:line:digitSCA_end} converges to a stationary point of the approximate problem \eqref{fixedWreform} (the proof is similar to Theorem~\ref{theo_conv_proof_WF}). Although the solution does not necessarily represent a stationary point of the original problem \eqref{fixedW}, one can claim that the solution is close to a stationary point of \eqref{fixedW} if the Neumann series approximation is tight. Algorithm~\ref{alg:BDRIS} ensures the tightness of the Neumann series approximation by setting $\gamma \ll 1$ in \ref{eq:tau_update}, which leads to satisfying \eqref{eq:deltacondition} \cite{direnzomutual2, direnzoclerckxmutual, stewart1998matrix}. Following Theorem~\ref{theo_conv_proof_WF}, Algorithm~\ref{alg:BDRIS} provides an approximate feasible solution for \eqref{graph_problem}.

\subsection{Extension to Scenarios with a Direct ET–ER Link}

We now discuss the extension of the proposed algorithms to scenarios where a direct transmitter-receiver link is present. Let us denote the direct component as ${h}_{D, n}$, leading to the cascade channel being formulated as $h_n = {h}_{D, n} + \mathbf{h}_{R,n}^T\mathbf{\Theta}{\mathbf{h}_{I,n}}$. This does not change the procedure of the waveform optimization algorithm since the cascade channel is an input to Algorithm~\ref{alg:waveformsca}. For SDR and SDP-based beamforming, we can rewrite \eqref{eq:hnreform2} as $h_n = g{h}_{D, n} + \mathbf{a}_n^T \boldsymbol{\theta}$, where $g$ is an auxiliary variable. Then, $h_n$ can be reformulated as $h_n = \mathbf{\hat{a}}_n^T \boldsymbol{\hat{\theta}}$, where $\boldsymbol{\hat{\theta}} = [\boldsymbol{\theta}, g]$ and $ \mathbf{\hat{a}}_n = [\mathbf{a}_n, {h}_{D, n}]$ are the concatenated vectors. Next, Algorithm~\ref{alg:SDPRIS} can be used to obtain $\boldsymbol{\theta}^*$, and the final solution can be computed by removing the last element in $\boldsymbol{\theta}^*/g$. Note that IT-BDRIS and SCA-BDRIS beamforming can be directly applied to the cases with a direct link only by substituting the new cascade channel (with a direct link).

\begin{algorithm}[t]
	\caption{Iterative beamforming and waveform optimization for BD-RIS-aided WPT relying on SCA, KKT, and the Neumann approximation (IT-BDRIS).} \label{alg:BDRIS}
	\begin{algorithmic}[1]
            \State \textbf{Input:} $\mathcal{G}$, $\beta$, $\upsilon$, $\gamma$, $\mathbf{h}_{R,n}, \mathbf{h}_{I,n}, \forall n$ \quad \textbf{Output:} $s_n^{(l)}, \forall n$, $\mathbf{Z}^{(l)}$
            \State \textbf{Initialize:} 
            \State $f^\star =0$, $i_{dc} = \infty$ \label{alg2:line:init_start}, $[\mathbf{Z}^{(l)}]_{i,m} = jZ_0, \forall i,m$ 
            \State Compute $\mathbf{P}$, $h_n$, and $\tau^{(l)}$ using \eqref{eq:permute_matrix}, \eqref{eq:hnreform}, and \eqref{eq:tau_update}
            \State Initialize $s_n^{(l)}, \forall n$ using \eqref{eq:SMF}
            \Repeat\label{alg1:line:alter_start}
                \State \hspace{-2mm} $\xi_2 =\infty$, $i_{dc}^\star \leftarrow i_{dc}$
                \Repeat\label{alg1:line:digitSCA_start}
                    \State \hspace{-2mm} $\xi_2^\star \leftarrow \xi_2$, compute $\nu_r^{(l)}, \forall r$ using \eqref{eq:Z_nu_update}
                    \State \hspace{-2mm} Obtain $\mathbf{Z}^{(l + 1)}$ using \eqref{eq:Z_omega_star}, \eqref{eq:Z_omega_update}, and \eqref{eq:Z_Z_update}
                    \State \hspace{-2mm} Compute $\xi_2$ and $\tau^{(l + 1)}$ using \eqref{fixedWreforma} and \eqref{eq:tau_update}
                    \State \hspace{-2mm} $\mathbf{Z}^{(l)} \leftarrow \mathbf{Z}^{(l + 1)}$, $\tau^{(l)} \leftarrow \tau^{(l + 1)}$, $l \leftarrow l  +1$
                \Until{$\norm{1 - {\xi_2^\star}/{\xi_2}}\leq \upsilon$}\label{alg1:line:digitSCA_end}
                \State \hspace{-2mm} Compute $h_n$ using \eqref{eq:hnreform}
                \State \hspace{-2mm} Run Algorithm~\ref{alg:waveformsca} to update $s_n^{(l)}$, compute $i_{dc}$ using \eqref{eq:parseval}
            \Until{$\norm{1 - i_{dc}^\star/ i_{dc}}\leq \upsilon$}\label{alg1:line:alter_end}
\end{algorithmic} 
\end{algorithm}

\subsection{Complexity Analysis}\label{subsec:companal}

Herein, we discuss the time complexity of the proposed algorithms in detail:\\
\textit{IT-WF:} The time complexity at each optimization iteration is $\mathcal{O}(N)$, which corresponds to computing $\{s_n\}$.\\
\textit{SDR-BDRIS:} The number of variables in \eqref{SDP_problem} is ${\bar{M}(\bar{M} + 1)}/2$, while the rest of the entries in $\mathbf{X}$ are determined according to the Hermitian structure, and the size of these Hermitian matrix sub-space is ${\bar{M}}^2$. Additionally, the number of constraints scales with $\bar{M}^2$ in \eqref{SDP_problem}. Interior-point methods require $\mathcal{O}(\sqrt{n})$ iterations, where $n$ is the number of constraints, and each iteration costs $\mathcal{O}\!\big(\max\{n \bar M^{3},\, n^{2}\bar M^{2},\, n^{3}\}\big)$ \cite{boyd2004convex}. Thus, the per-iteration complexity is $\mathcal{O}(\bar M^{6})$, and the overall complexity grows polynomially with $\bar M$ in addition to the $\mathcal{O}(N)$ cost of IT-WF. Thus, SDR-BDRIS is significantly heavier than IT-BDRIS for large BD-RIS sizes.\\
\textit{IT-BDRIS:} In the most time-consuming case, i.e., fully connected impedance network,\footnote{Notice that the increased complexity of optimizing BD-RIS also leads to higher processing power consumption for computing the corresponding scattering matrix, which should be considered in practical implementations.} $\boldsymbol{\omega}$ includes $\bar{M}$ non-zero entries, leading to $\mathcal{O}(\bar{M})$ complexity for obtaining the solution and the dual variables. Combining this with the negligible complexity for obtaining $\mathbf{Z}^{(l + 1)}$ leads to $\mathcal{O}(\bar{M})$ time complexity for each optimization iteration. Thus, the proposed IT-BDRIS is much more efficient than the traditional solver-based SCA. \\
\textit{SCA-BDRIS:} Each inner problem in \eqref{scatter_probelm_SCA} involves an $M\times M$ matrix variable and a fixed small number of linear matrix inequalities; thus, the per-iteration complexity is $\mathcal{O}(M^{3})$. Since several outer SCA iterations are required, the overall complexity is polynomial in $M$, higher than IT-BDRIS but substantially lower than the SDR-based formulation.

\section{Numerical Analysis}\label{sec:numerical}

\begin{figure}[t]
    \centering
    \includegraphics[width=0.9\columnwidth]{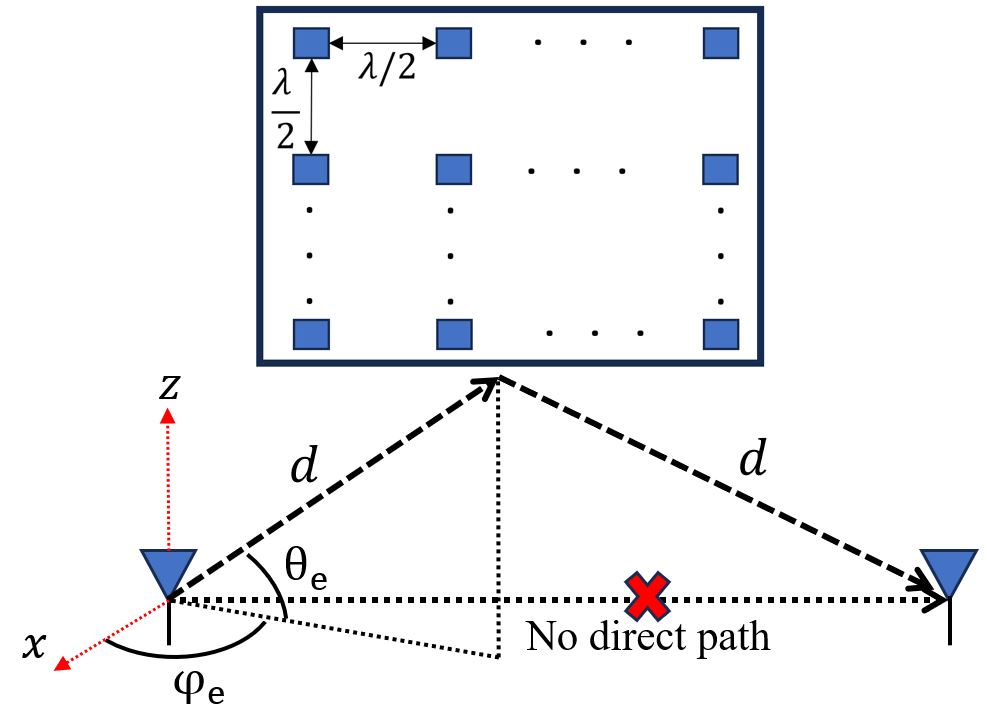} 
    \caption{The simulation layout.}
    \label{fig:layout}
\end{figure}

\begin{figure*}[t]
    \centering
    \includegraphics[width=0.245\textwidth]{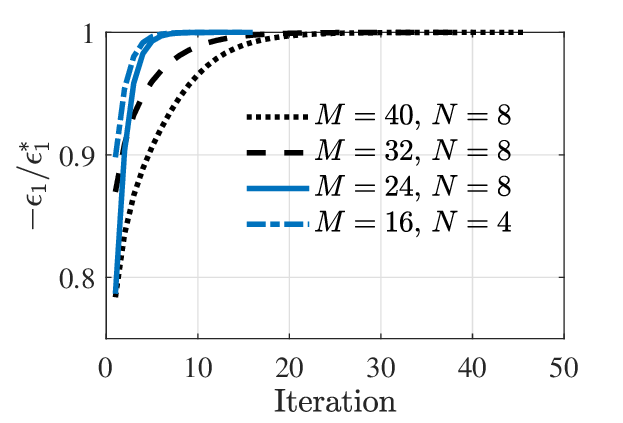}
    \includegraphics[width=0.245\textwidth]{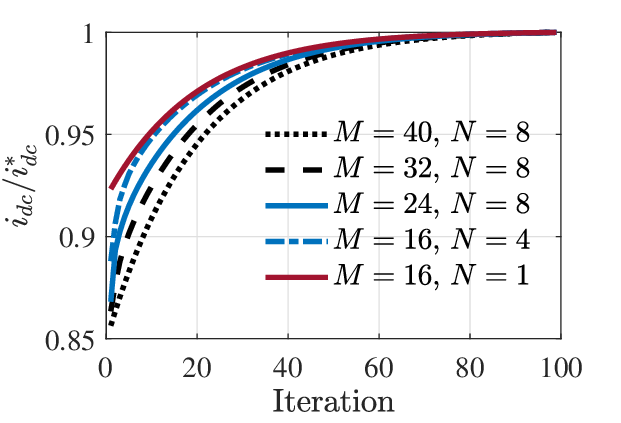} 
    \includegraphics[width=0.245\textwidth]{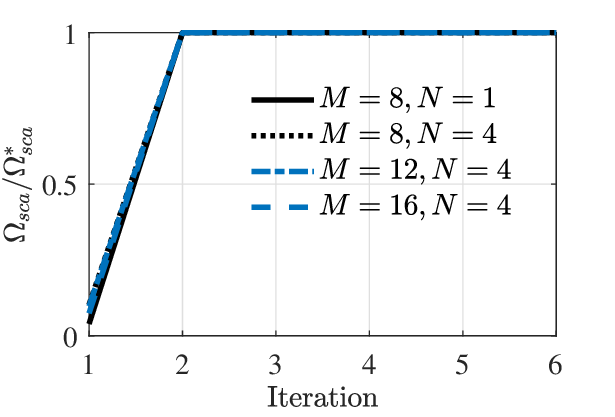}
    \includegraphics[width=0.245\textwidth]{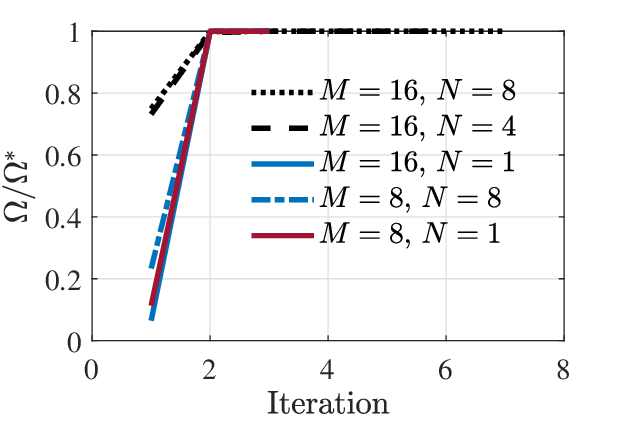}
    \caption{The convergence performance of (a) IT-WF (left), (b) IT-BDRIS (middle-left), (c) SCA-BDRIS (middle-right), (d) SDR-BDRIS (right). We adopt a random channel realization and assume a fully connected BD-RIS.}
    \label{fig:convergence}
\end{figure*}

In this section, we evaluate the system's performance in a WiFi-like scenario with $f_c = 2.4$~GHz, a bandwidth of $\textrm{BW} = 10$~MHz, and $\Delta f = \textrm{BW}/N$. We consider the path gain due to large-scale fading at a distance \( d \) to be \( L_0 d^{-2} \), where \( L_0 = -40~\mathrm{dB} \) is the reference path power gain at a 1~m distance. The BD-RIS elements are positioned on a uniform planar array with $\lambda/2$ as the distance between two adjacent elements, where $\lambda$ is the wavelength. For the sake of simplicity, we consider a symmetric simulation setup where the incident and reflective path lengths are $d = 2$~m, while $\theta_e = \phi_e = \pi/6$ are the elevation and azimuth angles. The simulation layout is illustrated in Fig.~\ref{fig:layout}. The number of iterations for Algorithm~\ref{alg:PS_obtain} and Gaussian randomization is set to 50000. In SCA-BDRIS, initial value of $\sigma$ is set to $10^{-5}$ and $\iota = 1$. Finally, the rectifier parameters are selected based on the circuit simulations in \cite{clerckx2018beneficial, clreckxWFdesign} for the small-signal regime such that $K_2 = 0.17$ and $K_4 = 957.25$.

Channels are modeled with quasi-static Rician fading as
\begin{equation}
    \mathbf{h}_n = \sqrt{\kappa/(\kappa + 1)}\mathbf{h}_n^{\text{LoS}} + \sqrt{1/(\kappa + 1)}\mathbf{h}_n^{\text{NLoS}},
\end{equation}
where $\kappa$ is the Rician factor and the $\mathbf{h}_n^{\text{LoS}}$ is modeled using the far-field LoS channel for uniform planar arrays with elevation and azimuth angles as inputs. The NLoS part is modeled with Rayleigh fading considering $L = 18$ delay taps with realizations following a circularly symmetric complex Gaussian distribution with a random power $p_l$, such that $\sum_{l = 1}^L p_l = 1$. Furthermore, we introduce an additional parameter $\alpha$ for the NLoS part such that the coherence bandwidth is $B_c = \alpha \textrm{BW}$. Hereby, the delay spread is $\varsigma = 1/B_c$, while the maximum delay spread is considered to be $\varsigma_m = 2 \varsigma$. Then, the tap delays $t_l, \forall l$ are generated uniformly spaced up to $\varsigma_m$, and the frequency response of each tap at the $n$th sub-carrier is multiplied by its corresponding delay term $e^{j2\pi f_n t_l}$. Hereby, selecting a small $\alpha$ such that $B_c$ is much smaller than $\textrm{BW}$ leads to more frequency selectivity of the channel. Nevertheless, a large $\alpha$ leads to $B_c \gg \textrm{BW}$, which means all delay taps can be received within the coherence bandwidth, leading to a frequency-flat channel. 

Note that the transmitter-RIS and RIS-receiver distances are chosen such that the far-field communications assumption holds. Specifically, the incident and reflective path lengths are larger than the Fraunhofer distance\cite{MyEBDMA}.\footnote{The Fraunhofer distance is formulated as $2D^2/\lambda$, where $D$ is the antenna diameter.}The results are averaged over 200 random channel realizations unless otherwise stated. To ensure the tractability of the nonlinear rectenna model and maintain the validity of the fourth-order approximation, we consider up to eight subcarriers, which is sufficient to capture the key performance trends and allow a fair comparison between D-RIS and BD-RIS architectures.

\subsection{Convergence Analysis}

Fig.~\ref{fig:convergence} illustrates the convergence performance of the IT-BDRIS, IT-WF, SCA-BDRIS, and SDR-BDRIS. It is seen that all algorithms iteratively converge toward a solution. Moreover, the number of required iterations increases with $M$ and $N$ as discussed in Section~\ref{subsec:companal}. Moreover,  Fig.~\ref{fig:convergence}.b-Fig.~\ref{fig:convergence}.d highlight that the proposed methods for alternating optimization of waveform and passive beamforming lead to convergence. Although the number of iterations in SDR-BDRIS is lower than in IT-BDRIS, each iteration of SDR-BDRIS requires solving multiple SDP problems with high complexity. At the same time, IT-BDRIS relies only on closed-form computations. Observe that the number of iterations for different setups is the same in SDR-BDRIS, but the time complexity of solving \eqref{SDP_problem} drastically increases with $M$ since the hermitian matrix subspace scales with $M^4$.


\subsection{Algorithms Performance Comparison}\label{subsec:converge}

Herein, we compare the performance of the proposed algorithms and consider also an optimal benchmark derived for $N = 1$. Note that for a single-carrier system ($N = 1$), problem \eqref{graph_problem} becomes equivalent to maximizing the RF power, i.e., the second-order term in $i_{dc}$, in terms of the solution. Thus, the framework introduced in \cite{BD-RIS_Graph} serves as a benchmark when $N=1$ and is labeled as optimal in Fig.~\ref{fig:bdris_compN14}.a.

Fig.~\ref{fig:bdris_compN14} compares the performance of the proposed algorithms for $N= 1$ and $N = 4$ as a function of $M$. Note that SDR-BDRIS achieves optimal performance, while SCA- and IT-BDRIS lead to relatively lower $i_{dc}$ values, but are still close to the optimal solution, and the gap increases with $M$. Moreover, for $N = 4$, the same trend is observed, where SDR-BDRIS outperforms the other proposed approaches. However, as stated earlier, the complexity of SDR-BDRIS is considerably higher than that of SCA-BDRIS and IT-BDRIS, especially for large values of $M$; thus, a trade-off exists between performance and complexity, which is observable in the figures. Moreover, SCA-BDRIS loses performance due to constraint approximation, which may lead to divergence from the feasible space over iterations. Note that there is no optimal benchmark for the multi-carrier case. Fig.~\ref{fig:bdris_comp_overN} compares the performance of the proposed algorithms as a function of $N$, showing that SDR-BDRIS achieves the best performance also in this case.

\begin{figure}[t]
    \centering
    \includegraphics[width=0.49\columnwidth]{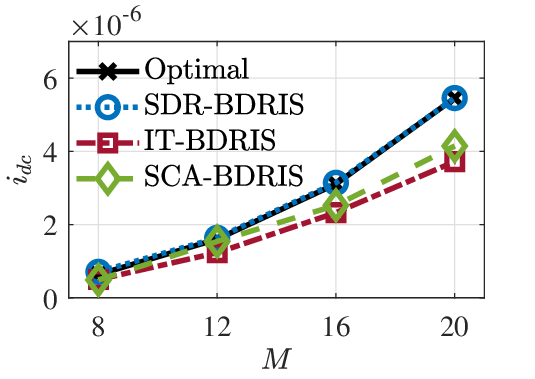} 
    \includegraphics[width=0.49\columnwidth]{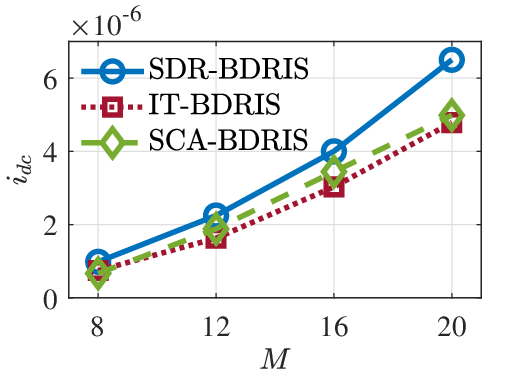}
    \caption{Average $i_{dc}$ at the ER as a function of $M$ using different algorithms for (a) $N = 1$ (left) and (b) $N = 4$ (right) with NLoS channel ($\alpha = 0.1$), $P_T = 50$~dBm, and a fully connected BD-RIS.}
    \label{fig:bdris_compN14}
\end{figure}

\begin{figure}[t]
    \centering
    \includegraphics[width=0.75\columnwidth]{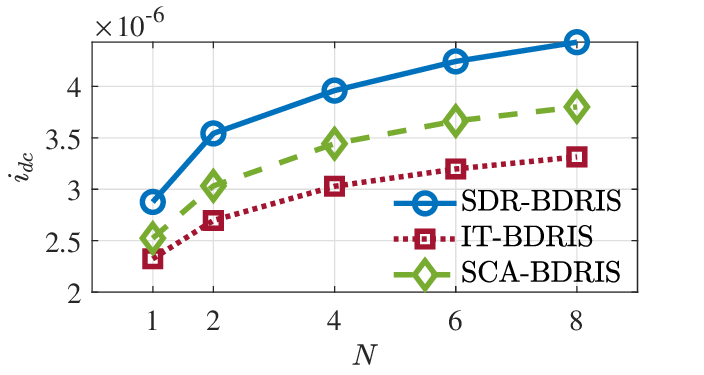} 
    \caption{Average $i_{dc}$ at the ER as a function of $N$ using different algorithms with $\alpha = 0.1$, $M = 16$, $P_T = 50$~dBm, and a fully connected BD-RIS.}
    \label{fig:bdris_comp_overN}
\end{figure}

\begin{table}[t]
\centering
\begin{tabular}{|c|c|c|c|c|}
\hline
\multirow{2}{*}{\textbf{Setup}} & \multicolumn{2}{c|}{\textbf{SDP-BDRIS}} & \multicolumn{2}{c|}{\textbf{SDR-BDRIS}} \\ 
\cline{2-5}
 & {$i_{dc}$ ($\mu$A)} & {DR} & {$i_{dc}$ ($\mu$A)} & {DR} \\ 
\hline
$M = 4,\ N = 8$ & 0.25 & 0.995 & 0.25 & 0.573 \\ 
\hline
$M = 8,\ N = 4$ & 1.04 & 0.998 & 1.04 & 0.57 \\ 
\hline
$M = 8,\ N = 8$ & 1.25 & 0.997 & 1.25 & 0.551 \\ 
\hline
$M = 12,\ N = 8$ & 2.08 & 0.996 & 2.08 & 0.684 \\ 
\hline
\end{tabular}
\caption{SDP-BDRIS and SDR-BDRIS performance comparison with 5 random realizations of Rayleigh fading channel.}
\label{tab:sdpsdr}
\end{table}

Table~\ref{tab:sdpsdr} compares the SDR-BDRIS and the method that models the rank-1 constraint using \eqref{eq:rank1atomic}, which is referred to as SDP-BDRIS.\footnote{Due to the high complexity of SDP-BDRIS, extensive simulations result in significant time overhead. Therefore, a brief comparison is presented to highlight the relative differences in the outcomes.} Herein, we illustrate the eigenvalue dominance ratio (DR) for each algorithm, which is computed by dividing the dominant eigenvalue of the obtained $\mathbf{X}$ by the sum of eigenvalues. It can be seen that the DR of SDP-BDRIS is much higher than SDR-BDRIS, ensuring a rank-1 solution. On the other hand, the DR values of SDR-BDRIS verify that it does not guarantee a rank-1 solution, and leveraging approximations to obtain a feasible solution is inevitable. Interestingly, it is observed that both algorithms lead to similar $i_{dc}$, and the rank-1 approximations of SDR-BDRIS only reduce its complexity and do not lead to performance losses.


\subsection{Waveform Behaviour \& Channel Analysis}

One of the key considerations of WPT with non-linear EH is the operating regime of the rectifier. Specifically, if the input RF power to the rectifier is relatively low, the RF-to-DC conversion efficiency drops significantly. On the other hand, if the RF input power is high such that the rectifier's diode enters the breakdown region, the mathematical model for the small-signal regime is no longer valid. Thus, it is important to select $P_T$ properly such that the received RF power is in the proper regime \cite{clerckx2018beneficial}. For this, we proceed with some discussions on the system's behavior for different $P_{T}$ values. 

Fig.~\ref{fig:poweralocation} presents the normalized gain of the channels and input signal for different $P_T$ and $\alpha$ values. Interestingly, we can see that the variation among the gains of the cascade channel at different sub-carriers becomes smaller as the incident and reflective channels become less frequency-selective (increasing $\alpha$), but obviously $P_T$ has no impact on it. This is caused by the nature of frequency-flat channels, where the channel values at different sub-carriers are almost equal, leading to the BD-RIS configuration impacting them similarly. However, when the channel becomes more frequency-selective, the BD-RIS configuration impacts the independent channels at different sub-carriers differently. Moreover, the power allocation pattern among the different sub-carriers of the input signal also changes with both $P_T$ and $\alpha$. Indeed, for low $P_T$, the whole power budget tends to be allocated to one or a few sub-carriers, especially as $\alpha$ decreases. The reason is that when $P_T$ is relatively small such that the input RF power to the rectifier is low, the rectifier operates in the low-efficiency regime. This leads to the fourth-order term of DC power being negligible and the second-order term becoming dominant. Thus, the designed waveform does not leverage all sub-carriers efficiently to generate a high peak-to-average power ratio (PAPR). On the other hand, when $P_T$ is chosen such that the RF power is in the desired range ($P_T = 50$~dBm in this case), the rectifier operates in the high-efficiency regime, and the fourth-order term becomes large enough to impact the DC harvested power leading to a more diverse pattern in power allocation among sub-carriers. Thus, $P_T$ needs to match properly with system parameters so that the rectifier non-linearity is leveraged efficiently. For instance, when $P_T = 50$~dBm, the power is allocated to the stronger sub-carriers in a frequency-selective scenario, while the power is allocated to all sub-carriers in a frequency-flat case with $\alpha = 10$. Motivated by this analysis, we consider $P_T = 50$~dBm in the rest of this section.

\begin{remark}
    The transmit power level of 50~dBm is chosen to ensure sufficient RF input at the rectifier under the single-antenna setup, which lacks spatial diversity gain. Although this power level is higher than typical practical values, it helps isolate the waveform design effects. In real-world systems with multi-antenna configurations, the required transmit power can be significantly reduced due to spatial beamforming and diversity gains.
\end{remark}

\begin{figure*}[t]
    \centering
    \includegraphics[width=0.32\textwidth]{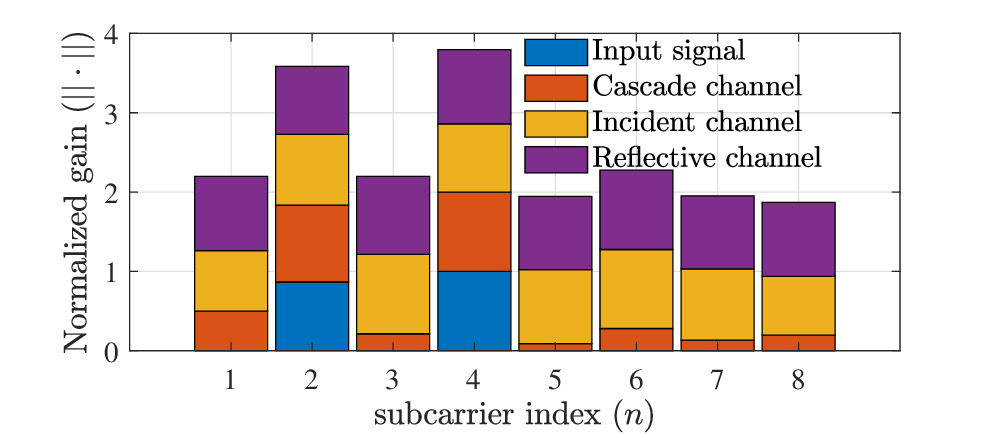}
    \includegraphics[width=0.32\textwidth]{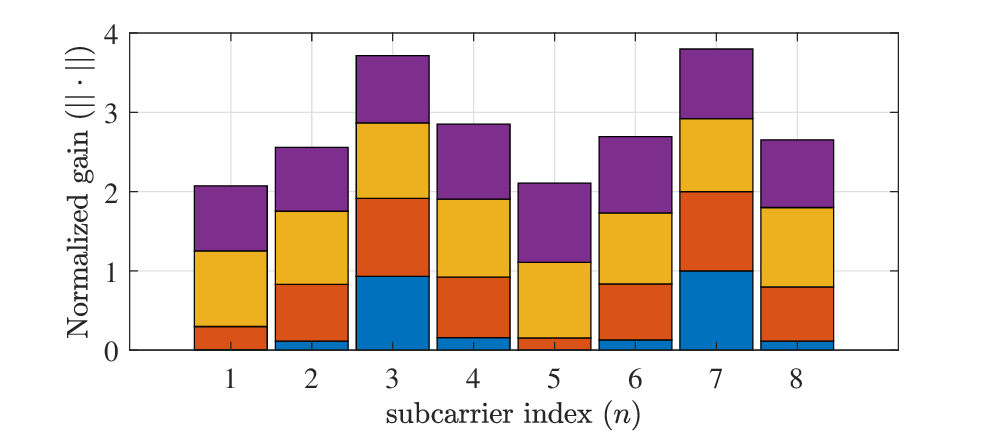}
    \includegraphics[width=0.32\textwidth]{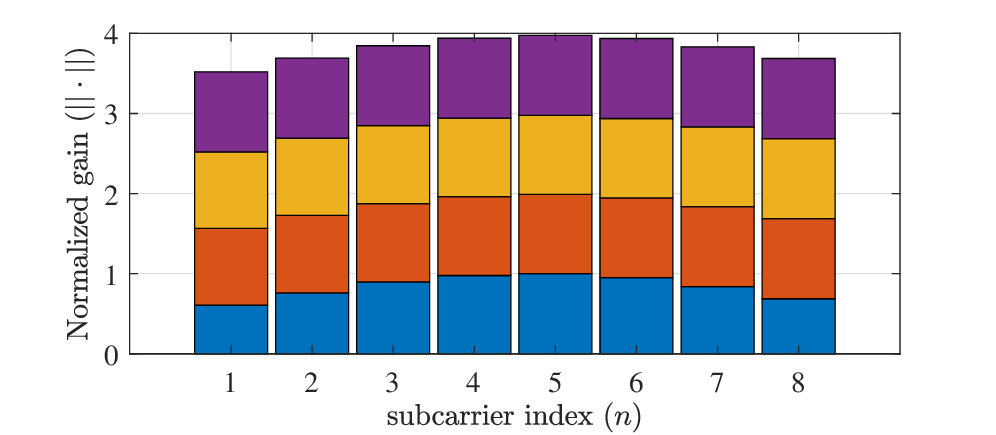} \\
    \includegraphics[width=0.32\textwidth]{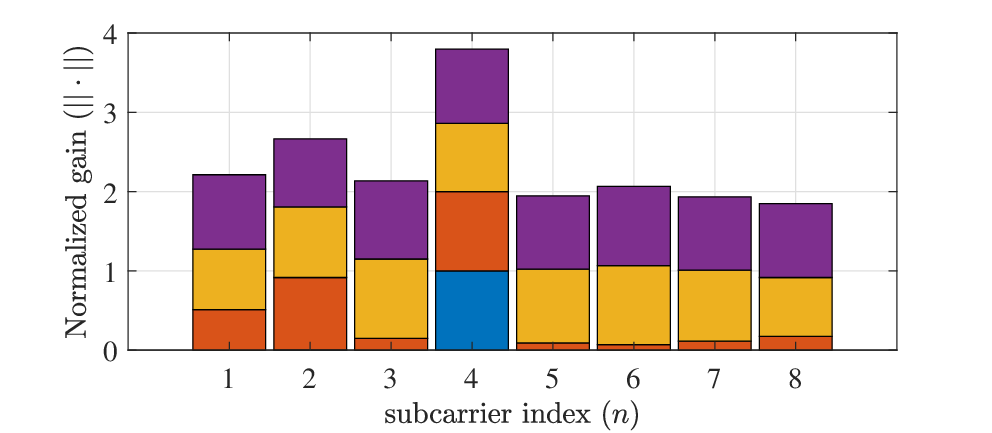}
    \includegraphics[width=0.32\textwidth]{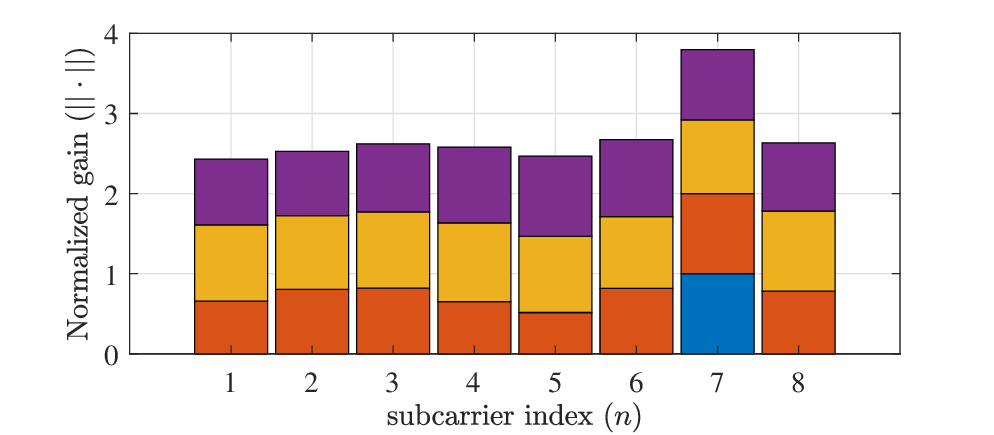}
    \includegraphics[width=0.32\textwidth]{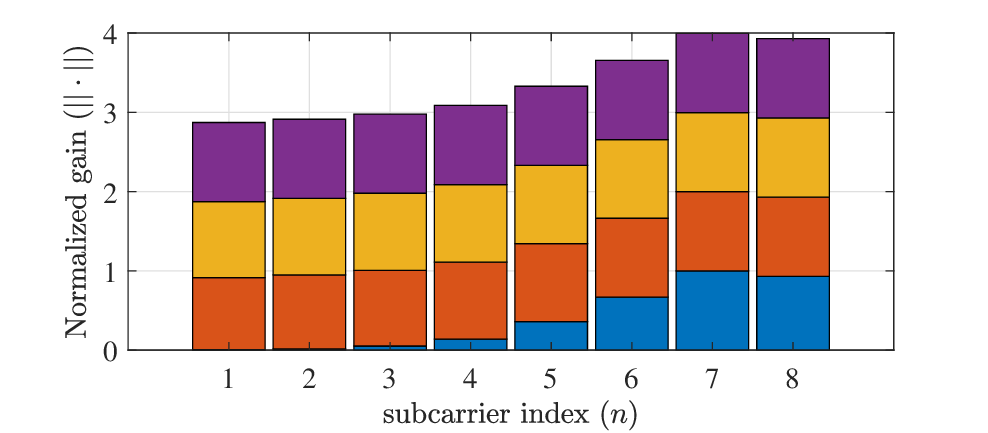}
    \caption{The normalized gain of the channels and input signal at different sub-carriers for (a) $P_T = 50$~dBm (top) and (b) $P_{T} = 30$~dBm (bottom) for $\alpha = 0.1$ (left), $\alpha = 1$ (middle), and $\alpha = 10$ (right). We assume a fully connected BD-RIS with $M=32$, $N=8$, and a pure NLoS ($\kappa = 0$) random channel realization. Note that the bars are represented in a stacked fashion, meaning that bars with the same length represent the same normalized gain.}
    \label{fig:poweralocation}
\end{figure*}

Fig.~\ref{fig:timesignal} shows the time domain signal for the same channel realizations as in Fig.~\ref{fig:poweralocation}. Note that higher PAPR signals are received at the ER for $P_T=50$~ dBm to leverage the rectifier's non-linearity efficiently, even for lower $\alpha$. On the other hand, for low $P_T$, the received signal has  PAPR $\approx 0$~dB for small $\alpha$ and lower PAPR compared to $P_T = 50$~dBm for large values of $\alpha$. Thus, as earlier stated, high PAPR signals are beneficial to leverage the rectifier's non-linearity, but only when the input RF power is in the high-efficiency regime of the rectifier.

\begin{figure}[t]
    \centering
    \includegraphics[width=0.48\columnwidth]{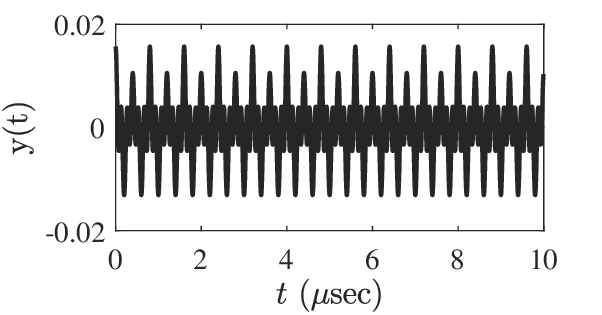}
    \includegraphics[width=0.48\columnwidth]{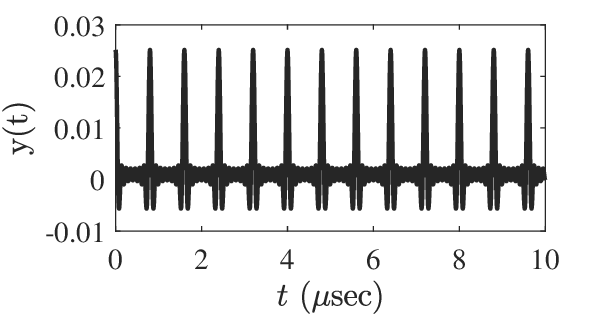} \\
    \includegraphics[width=0.48\columnwidth]{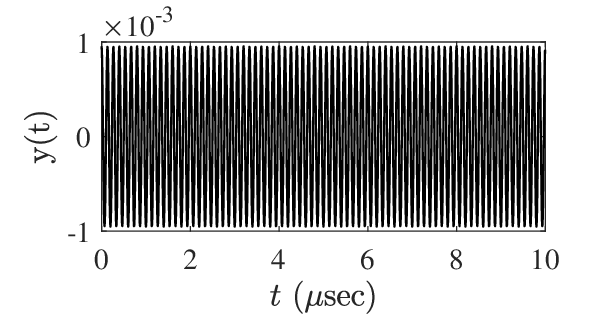}
    \includegraphics[width=0.48\columnwidth]{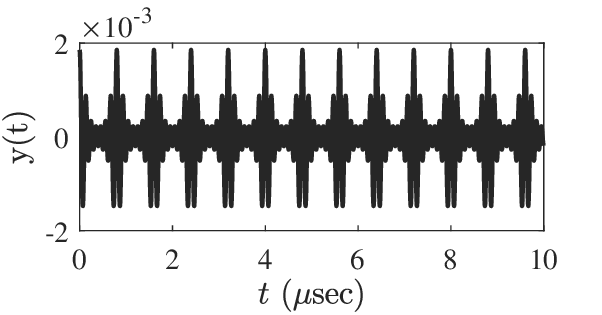}
    \caption{The received signal in time domain for (a) $P_T = 50$~dBm (top) and (b) $P_{T} = 30$~dBm (bottom) for $\alpha = 1$ (left) and $\alpha = 10$ (right) with fully connected BD-RIS, $M=32$, and $N=8$ at a random channel realization.}
    \label{fig:timesignal}
\end{figure}


\subsection{D-RIS versus BD-RIS}

Herein, we compare the performance of D-RIS and fully connected BD-RIS. The BD-RIS results are obtained using the SDR-BDRIS approach, while for D-RIS, this approach can be straightforwardly modified as presented in \cite{clerckxRISWPT}. 

Note that for far-field LoS, the channel is deterministic and is computed using the relative phase difference of elements. Thus, we assume that the variation among the sub-carriers is negligible such that for a given $\boldsymbol{\Theta}$, $h_n = h$. Hereby, \eqref{eq:parseval} can be rewritten as $\alpha_1 \norm{h}^2 + \alpha_2 \norm{h}^4$, where $\alpha_1$ and $\alpha_2$ are positive coefficients depending on the waveform and rectenna circuit. To maximize this, it is sufficient to maximize the cascade channel gain $\norm{h}$, which can be achieved by constructing a diagonal $\boldsymbol{\Theta}$ such that $[\boldsymbol{\Theta}]_{i, i} = e^{-j([\tilde{\mathbf{h}}_{R,n}]_i + [\tilde{\mathbf{h}}_{I,n}]_i)}, \forall i$. Such a scattering matrix can be designed using D-RIS; thus, D-RIS can achieve the optimal performance and BD-RIS does not introduce any additional gain. Moreover, for $N = 1$, \eqref{graph_problem} becomes equivalent to maximizing the RF power, which allows for extending the findings of \cite{BD-RIS_scattering_clerckx} to single-carrier WPT. Thus, for a channel with NLoS components, we expect that BD-RIS outperforms D-RIS even with a continuous wave. 

Fig.~\ref{fig:DBDoverM} illustrates the average $i_{dc}$ as a function of $M$ for D-RIS and BD-RIS. It is seen in Fig.~\ref{fig:DBDoverM}.a that for far-field LoS channels, BD-RIS achieves the same performance as D-RIS in both single-carrier and multi-carrier systems, which complies with the mathematical analysis. However, when the channel tends to become frequency-selective, as in Fig.~\ref{fig:DBDoverM}.b, BD-RIS can leverage the extra degrees of freedom to impact different channel components effectively such that $i_{dc}$ becomes higher compared to D-RIS. Fig.~\ref{fig:DBDoverN} verifies the same pattern over $N$ with a given $M$. Specifically, BD-RIS and D-RIS achieve the same performance for any $N$ under LoS, while for the Rician channel with NLoS components, the BD-RIS outperforms D-RIS and the performance gap increases with $N$.  

\begin{figure}[t]
    \centering
    \includegraphics[width=0.49\columnwidth]{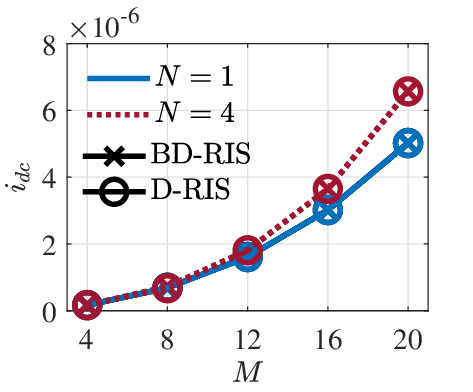}
    \includegraphics[width=0.49\columnwidth]{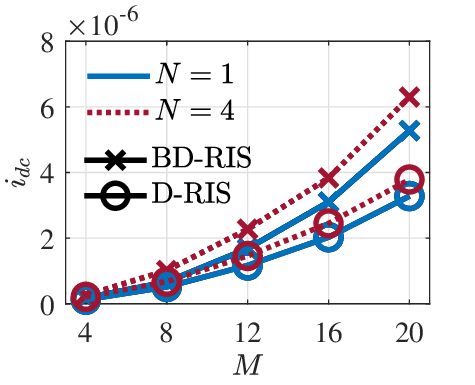} 
    \caption{Average $i_{dc}$ as a function of $M$ for (a) LoS (left) and (b) Rician channel with $\kappa = 0$~dB (right).}
    \label{fig:DBDoverM}
\end{figure}

\begin{figure}[t]
    \centering
    \includegraphics[width=0.49\columnwidth]{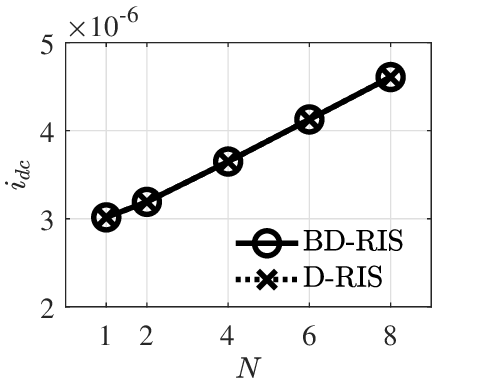} 
    \includegraphics[width=0.49\columnwidth]{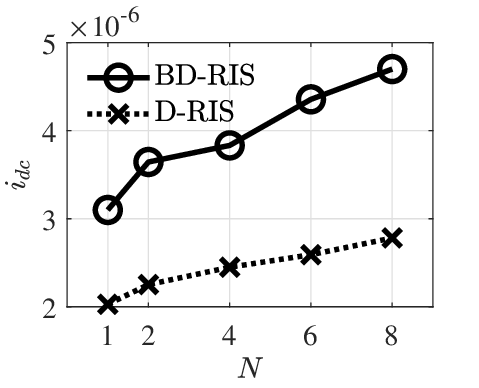}
    \caption{Average $i_{dc}$ as a function of $N$ for (a) LoS (left) and (b) Rician channel with $\kappa = 0$~dB (right) for $M = 16$.}
    \label{fig:DBDoverN}
\end{figure}

Fig.~\ref{fig:DBDchan} illustrates the normalized gain of the channels and the input signal at different sub-carriers. It is observed that the configuration of D-RIS/BD-RIS impacts the sub-carriers similarly under pure LoS, leading to small variations in the cascade channel. Thus, D-RIS provides sufficient degrees of freedom to compensate for the channels. On the other hand, for NLoS, BD-RIS leverages its additional degrees of freedom to create more dominant sub-carriers compared to D-RIS, leading to extra performance gains for frequency-selective channels. 

\begin{remark}
    Note that the performance gap might be different in the case of including power consumption-related constraints into the optimization problem in \eqref{graph_problem}. Specifically, this introduces trade-offs between circuit complexity and performance in BD-RIS, leading to a change of performance depending on the power consumption \cite{energy_BDRIS}. Notably, the fundamental Pareto tradeoff between performance and hardware complexity has been studied in \cite{Pareto_Frontier_nerini, BD-RIS_Graph}, and low-complexity optimal architectures exist, such as tree-connected and band/stem-connected, which can achieve the same performance as fully connected, but with lower hardware complexity \cite{Pareto_Frontier_nerini, BD-RIS_Graph}.
\end{remark}

\begin{figure*}[t]
    \centering
    \includegraphics[width=0.95\columnwidth]{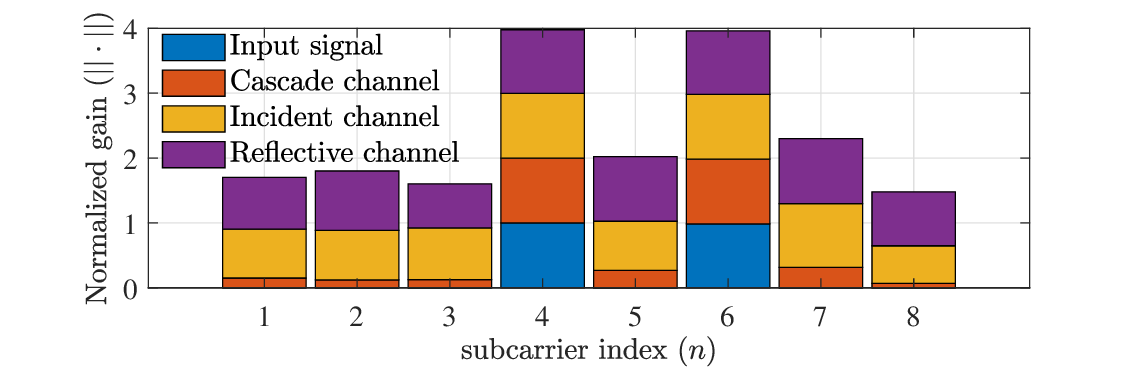} 
    \includegraphics[width=0.95\columnwidth]{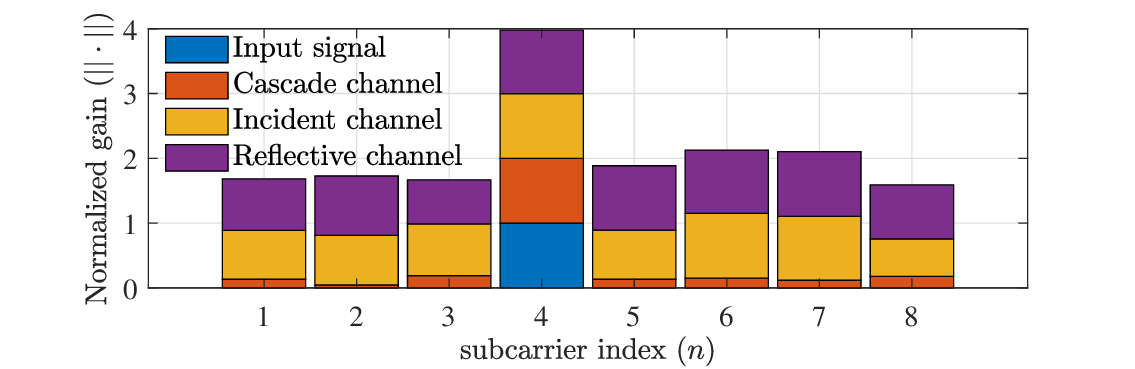} \\
    \includegraphics[width=0.95\columnwidth]{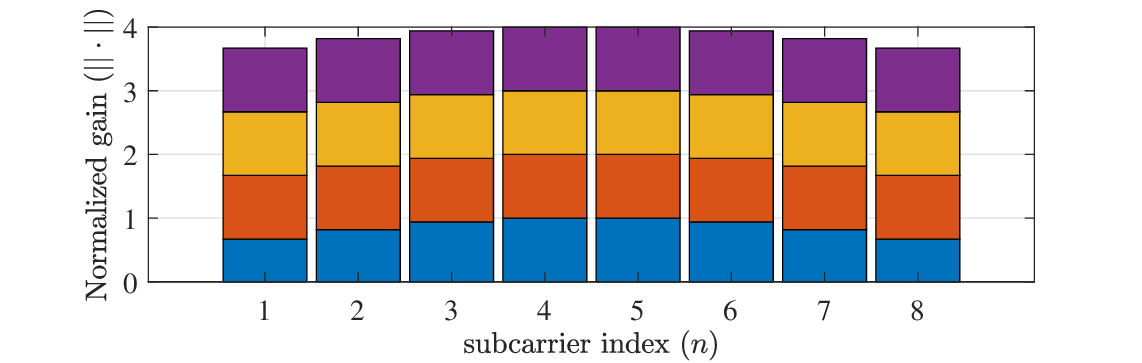} 
    \includegraphics[width=0.95\columnwidth]{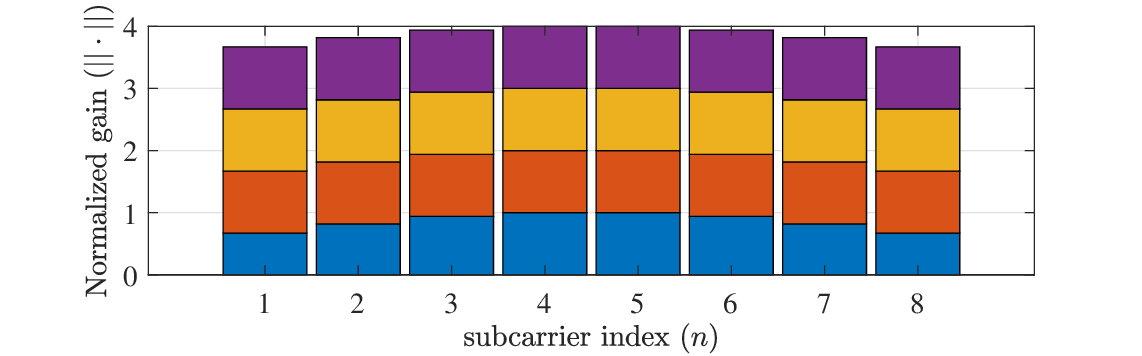} 
    \caption{The normalized gain of the channels and input signal at different sub-carriers for (a) NLoS channel (top) and (b) LoS channel (bottom) for BD-RIS (left) and D-RIS (right). We assume a random channel realization with $M = 16$.}
    \label{fig:DBDchan}
\end{figure*}

\section{Conclusions and Future Works}\label{sec:conclude}

We considered a BD-RIS-aided SISO WPT system with EH non-linearity. Moreover, we formulated a joint beamforming and waveform optimization problem aiming to maximize the harvested power at the ER. To address this problem, we decoupled the waveform optimization and beamforming problems and leveraged alternating optimization. We proposed an efficient iterative method for waveform optimization and four beamforming approaches, namely SDR-BDRIS, SDP-BDRIS, SCA-BDRIS, and IT-BDRIS. Furthermore, IT-BDRIS has the lowest computational complexity but at the cost of performance losses. Meanwhile, it was seen that SDR-BDRIS achieves the same performance as SDP-BDRIS with less computational complexity. The simulation results proved that the proposed algorithms converge to a local optimal solution, while the complexity scales with the number of RIS elements and the number of sub-carriers. Moreover, SDR-BDRIS  achieves the optimal performance for a single-carrier system and outperforms IT-BDRIS in other scenarios. Interestingly, we proved that BD-RIS achieves the same performance as D-RIS under LoS conditions (in the absence of mutual coupling), while it outperforms D-RIS in NLoS cases. Furthermore, our findings demonstrated that the frequency selectivity of the channel impacts the way that BD-RIS shapes the cascade channel. Finally, we showed that the waveform power allocation to different sub-carriers depends on the transmit power and the rectifier's operating regime.

As a prospect for future research, one can consider the following directions. The power consumption model of the BD-RIS, which depends on its architecture and the interconnections between different ports, can be analyzed to explore the trade-offs between complexity, power consumption, and performance. Moreover, the power consumption model can also include the complexity of the optimization algorithms as an additional factor, since the algorithms' complexity increases with the number of connections between different ports of BD-RIS. While this work provides a first step toward understanding the benefits of BD-RIS in WPT with nonlinear energy harvesting, extending the framework to more general setups is a promising direction for future research. In particular, incorporating multi-antenna transmitters and receivers or supporting multiple energy receivers would enable spatial multiplexing and further enhance the power transfer efficiency.

\appendices

\section{Proof of Proposition~\ref{theorem:vecotorization2}}\label{appen3}
We proceed by rewriting $h_n$ as 
\begin{align}
    h_n &= \mathbf{h}_{R,n}^T\boldsymbol{\Theta}{\mathbf{h}_{I,n}} 
    = \mathrm{Tr}(\mathbf{h}_{R,n}^T\boldsymbol{\Theta}{\mathbf{h}_{I,n}})  \nonumber \\
    &\labelrel={myinlab:1} \mathrm{Tr}({\mathbf{h}_{I,n}}\mathbf{h}_{R,n}^T\boldsymbol{\Theta})  = \mathrm{Tr}(\mathbf{H}_n\boldsymbol{\Theta}) \nonumber \\
    &\labelrel={myinlab:2} {\mathrm{Vec}(\mathbf{H}_n)}^T\mathrm{Vec}(\boldsymbol{\Theta}) 
    \labelrel={myinlab:3}{\mathrm{Vec}(\mathbf{H}_n)}^T\mathbf{P}\boldsymbol{\theta}
    = \mathbf{a}_n^T \boldsymbol{\theta}.
\end{align}
Assume $\mathbf{D}$, $\mathbf{F}$, and $\mathbf{H}$ are arbitrary matrices. Hereby, \eqref{myinlab:1} and \eqref{myinlab:2} come from $\mathrm{Tr}(\mathbf{DFH}) = \mathrm{Tr}(\mathbf{HDF})$ and $\mathrm{Tr}(\mathbf{D}^T\mathbf{F}) = \mathrm{Vec}(\mathbf{D})^T\mathrm{Vec}(\mathbf{F})$, respectively. Moreover, by defining $\boldsymbol{\theta} \in \mathbb{C}^{\bar{M} \times 1}$ as the vector containing the lower/upper-triangle elements in $\mathbf{\Theta}$, one can design a permutation matrix $\mathbf{P} \in \{0, 1\}^{M^2\times \bar{M}}$ such that $\mathbf{P}\boldsymbol{\omega} = \mathrm{Vec}(\mathbf{\Omega})$, which leads to \eqref{myinlab:3}.  \qed{}

\section{Proof of Proposition~\ref{theorem:2}}\label{appen2}
The Neumann approximation can be used to approximate the inverse of a matrix using another close matrix. Imagine $\mathbf{A}$ is a close matrix to an invertible $\mathbf{X}$, such that \cite{mimoneumann}
\begin{equation}
    \lim_{n \rightarrow \infty} (\mathbf{I} - \mathbf{A}\mathbf{X}^{-1})^n = 0.
\end{equation}
Then, the inverse of $\mathbf{A}$ can be written using the Neumann series as \cite{stewart1998matrix}
\begin{equation}\label{eq:neumanbase}
    \mathbf{A}^{-1} \approx \sum_{i = 0}^{\infty} \bigl(\mathbf{X}^{-1}(\mathbf{X} - \mathbf{A})\bigr)^i \mathbf{X}^{-1}.
\end{equation}
Now, let us consider $\mathbf{\Omega}$ as a small increment to the impedance matrix $\mathbf{Z}$, such that $\norm{[\mathbf{\Omega}]_{i,m}} \leq \delta, \forall i,m$, leading to $\mathbf{Z} + \mathbf{\Omega}$ being close to $\mathbf{Z}$. Hereby, we can write the first-order Neumann series as in \eqref{eq:neuman_impedance}. Note that the first-order approximation might lead to losing some information about higher-order terms. However, it has been proven that \eqref{eq:neuman_impedance} is sufficiently accurate if $\delta$ satisfies \eqref{eq:deltacondition} \cite{direnzomutual2, direnzoclerckxmutual, stewart1998matrix},
leading to the error of the approximation approaching zero. \qed{}

\bibliographystyle{ieeetr}
\bibliography{ref_abbv}

\begin{IEEEbiography}[{\includegraphics[width=1in,height=1.25in,clip,keepaspectratio]{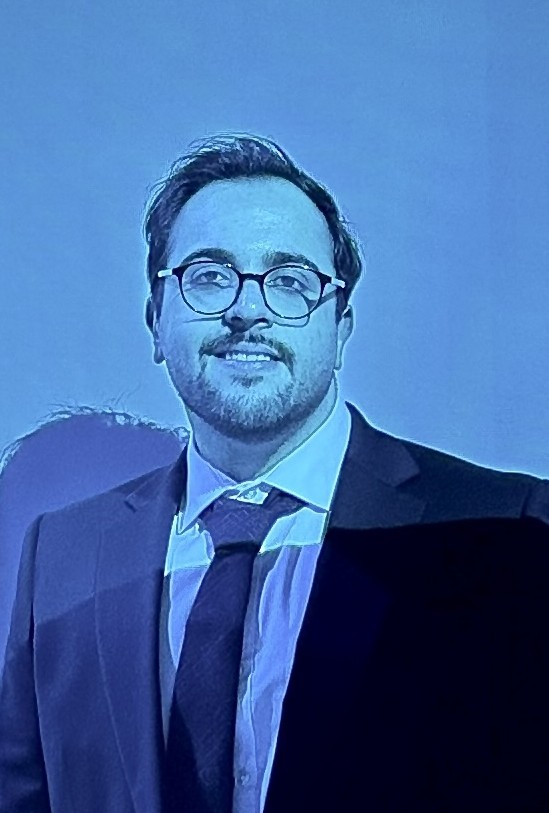}}]{Amirhossein Azarbahram}~(Graduate Student~Member, IEEE) received the B.Sc. degree in telecommunications from KN Toosi University of Technology, Iran, in 2020, and the M.Sc. degree in communications systems from Sharif University of Technology, Iran, in 2022. He is currently a Doctoral researcher at the Centre for Wireless Communications (CWC), University of Oulu, Finland. He was a visiting researcher at the Connectivity Section (CNT) of the Department of Electronic Systems, Aalborg University (AAU), Denmark, in 2024. His research interests include wireless communications and AI-enabled integrated sensing and communications. 
\end{IEEEbiography}

\begin{IEEEbiography}[{\includegraphics[width=1in,height=1.25in,clip,keepaspectratio]{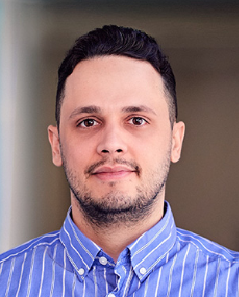}}]{Onel L. A. López}~(Senior~Member, IEEE) is an Associate Professor (tenure track) in sustainable wireless communications engineering at the University of Oulu, Finland.  He received the B.Sc. (1st class honors, 2013), M.Sc. (2017), and D.Sc. (with distinction, 2020) degree in Electrical Engineering from the Central University of Las Villas (Cuba), the Federal University of Paraná (Brazil), and the University of Oulu (Finland), respectively. In 2013-2015, he served as a telematics specialist at the Cuban telecommunications company (ETECSA). In 2020, he was a post-doctoral researcher in a joint project between the University of Oulu and Nokia Oulu, Finland. He was on a six-month research visit to Rice University and the University of Houston, Texas, US, in 2024. He is a collaborator to the 2016 Research Award given by the Cuban Academy of Sciences, a co-recipient of the 2019 and 2023 IEEE European Conference on Networks and Communications (EuCNC) Best Student Paper Award, and the recipient of both the 2020 Best Doctoral Thesis award granted by Academic Engineers and Architects in Finland TEK and Tekniska Föreningen i Finland TFiF in 2021 and the 2022 Young Researcher Award in the field of technology in Finland. He is currently an Associate Editor of the IEEE Transactions on Communications, IEEE Wireless Communications Letters, and IEEE Communications Letters. His research interests include sustainable IoT, energy harvesting, wireless RF energy transfer, wireless connectivity, machine-type communications, and cellular-enabled sensing and positioning systems.
\end{IEEEbiography}

\begin{IEEEbiography}[{\includegraphics[width=1in,height=1.25in,clip,keepaspectratio]{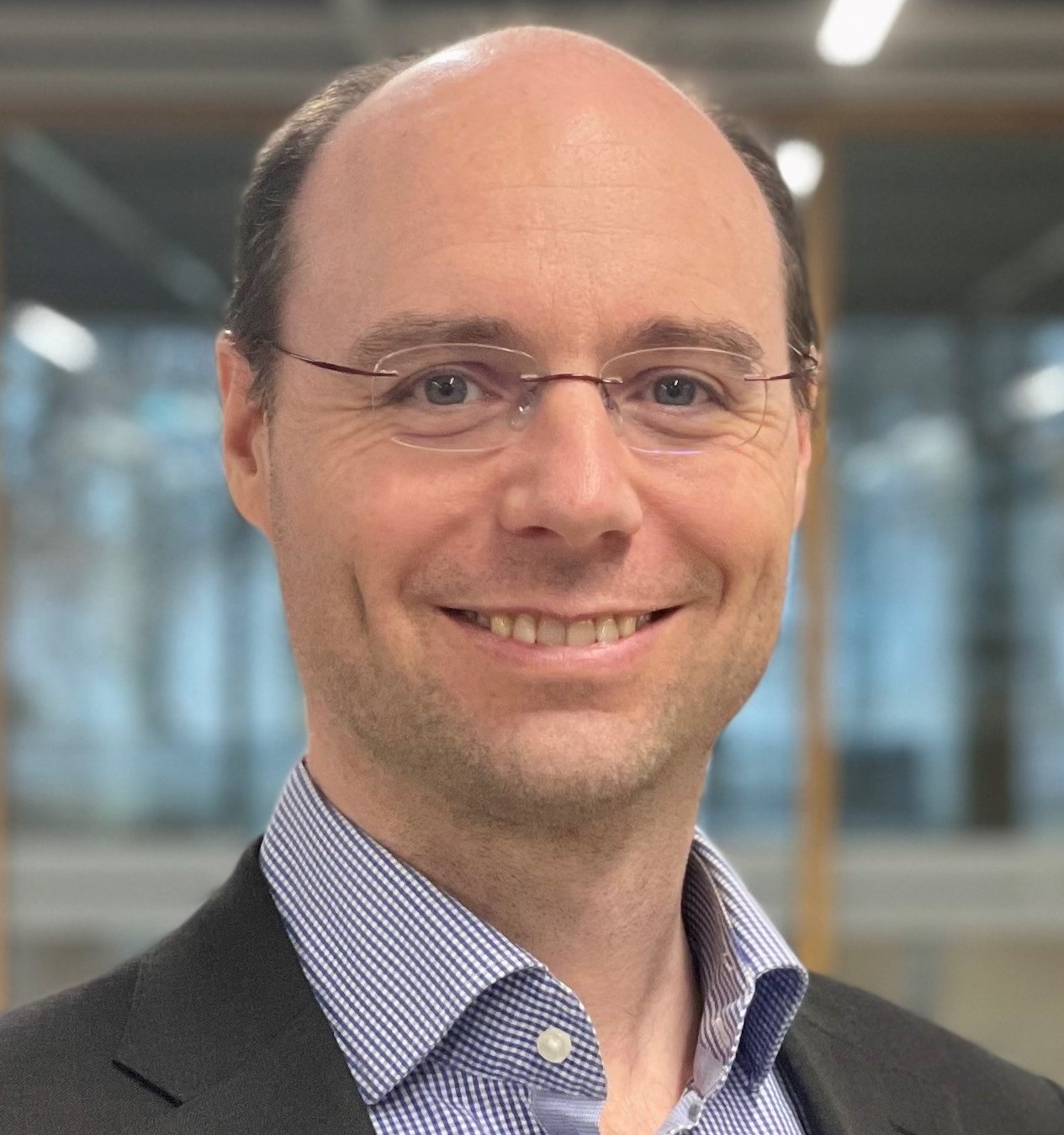}}]{Bruno~Clerckx}~(Fellow, IEEE) is a (Full) Professor, the Head of the Communications and Signal Processing Group, and the Head of the Wireless Communications and Signal Processing Lab, within the Electrical and Electronic Engineering Department, Imperial College London, London, U.K. He received the MSc and Ph.D. degrees in Electrical Engineering from Université Catholique de Louvain, Belgium, and the Doctor of Science (DSc) degree from Imperial College London, U.K. He spent many years in industry with Silicon Austria Labs (SAL), Austria, where he was the Chief Technology Officer (CTO) responsible for all research areas of Austria's top research center for electronic based systems and with Samsung Electronics, South Korea, where he actively contributed to 4G (3GPP LTE/LTE-A and IEEE 802.16m). He has authored two books on “MIMO Wireless Communications” and “MIMO Wireless Networks”, over 350 peer-reviewed international research papers, and 150 standards contributions, and is the inventor of 80 issued or pending patents among which several have been adopted in the specifications of 4G standards and are used by billions of devices worldwide. His research spans the general area of wireless communications and signal processing for wireless networks. He received the prestigious Blondel Medal 2021 from France for exceptional work contributing to the progress of Science and Electrical and Electronic Industries, the 2021 Adolphe Wetrems Prize in mathematical and physical sciences and the 2024 Georges Vanderlinden Prize in Electromagnetism and Telecommunications from Royal Academy of Belgium, multiple awards from Samsung, IEEE best student paper award, and the EURASIP (European Association for Signal Processing) best paper award 2022. He is a Fellow of the IEEE and the IET.
\end{IEEEbiography}

\begin{IEEEbiography}[{\includegraphics[width=1in,height=1.25in,clip,keepaspectratio]{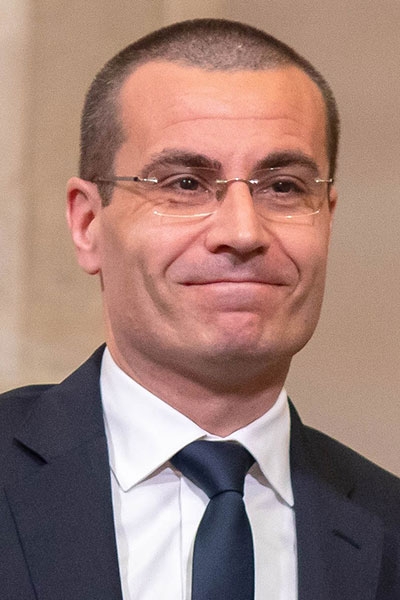}}]{Marco~Di~Renzo}~(Fellow, IEEE)  received the Laurea
(cum laude) and Ph.D. degrees in electrical engineering from the University of
L’Aquila, Italy, in 2003 and 2007, respectively, and the Habilitation à Diriger des
Recherches (Doctor of Science) degree from University Paris-Sud (currently Paris-Saclay University), France, in 2013. Currently, he is a CNRS Research Director
(Professor) and the Head of the Intelligent Physical Communications group with
the Laboratory of Signals and Systems (L2S) at CNRS \& CentraleSupelec, Paris-Saclay
University, Paris, France. Also, he is Chair Professor in Telecommunications
Engineering, the Director of the Centre for Telecommunications Research, and
the Head of the Telecommunications Group, Department of Engineering, King’s
College London, London, United Kingdom. He was a France-Nokia Chair of Excellence
in ICT at the University of Oulu (Finland), a Tan Chin Tuan Exchange Fellow
in Engineering at Nanyang Technological University (Singapore), a Fulbright Fellow
at The City University of New York (USA), a Nokia Foundation Visiting Professor at
Aalto University (Finland), and a Royal Academy of Engineering Distinguished Visiting
Fellow at Queen’s University Belfast (U.K.). He is a Fellow of the IEEE, IET, EURASIP,
and AAIA; an Academician of AIIA; an Ordinary Member of the European
Academy of Sciences and Arts, an Ordinary Member of the Academia Europaea,
and an Ordinary Member of the Italian Academy of Technology and Engineering; an
Ambassador of the European Association on Antennas and Propagation; and a
Highly Cited Researcher. His recent research awards include the Michel Monpetit
Prize conferred by the French Academy of Sciences, the IEEE Communications
Society Heinrich Hertz Award, and the IEEE Communications Society Marconi
Prize Paper Award in Wireless Communications. He served as the Editor-in-Chief of
IEEE Communications Letters from 2019 to 2023. Currently, he is a Voting Member
of the Fellow Evaluation Standing Committee, the Chair of the Publications Misconduct
Ad Hoc Committee, and the Director of Journals of the IEEE Communications
Society. Also, he is on the Editorial Board of the Proceedings of the IEEE.
\end{IEEEbiography}

\begin{IEEEbiography}[{\includegraphics[width=1in,height=1.25in,clip,keepaspectratio]{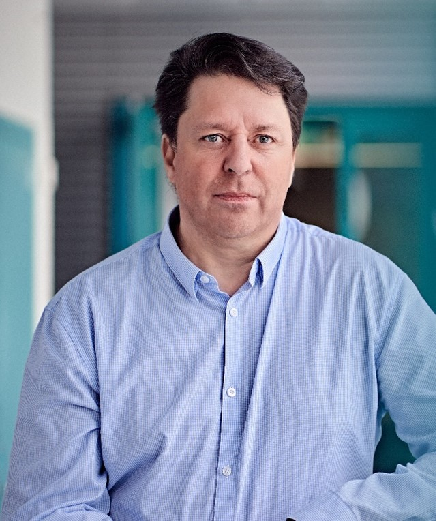}}]{Matti~Latva-aho}~(Fellow, IEEE) is a distinguished expert in wireless communications. He holds M.Sc., Lic.Tech., and Dr.Tech. (Hons.) degrees in Electrical Engineering from the University of Oulu, Finland, awarded in 1992, 1996, and 1998, respectively. From 1992 to 1993, he worked as a Research Engineer at Nokia Mobile Phones in Oulu before joining the Centre for Wireless Communications (CWC) at the University of Oulu. Prof. Latva-aho served as Director of CWC from 1998 to 2006 and later as Head of the Department of Communication Engineering until August 2014. He was nominated as an Academy Professor by the Academy of Finland in 2017. He is a Professor of Wireless Communications at the University of Oulu and served as Director of the National 6G Flagship Programme. He is also a Global Fellow at The University of Tokyo. In 2025, he was appointed Vice-Rector for Research at the University of Oulu for a five-year term. With an extensive portfolio of over 600 conference and journal publications, Prof. Latva-aho has significantly advanced the field of wireless communications. His contributions were recognized in 2015 when he received the prestigious Nokia Foundation Award for his groundbreaking research in mobile communications.
\end{IEEEbiography}

\end{document}